\documentclass[final,twoside,11pt]{entics} 
\usepackage{enticsmacro}
\usepackage{amsmath}
\usepackage{amsfonts}
\usepackage{amssymb}
\usepackage{mathtools}
\usepackage{mathrsfs}
\usepackage{stmaryrd}
\usepackage{nicefrac}
\usepackage{graphicx}
\usepackage[all]{xy}
\usepackage{proof}
\usepackage{caption}
\usepackage{ccicons}
\mathchardef\dash="2D
\newcommand{\V}{\mathcal{V}}
\newcommand{\N}{\mathbb{N}}

\newcommand{\R}{\mathbb{R}}

\newcommand{\defeq}{\triangleq}

\newcommand{\g}[1]{!_{#1} \,}

\newcommand{\eps}[1][n]{\epsilon^{#1}}  
\DeclarePairedDelimiter\norm{\lVert}{\rVert}%
\newcommand{\perm}{\mathbin{\simeq_\pi}}
\DeclareMathOperator{\sym}{\mathrm{Sym}}
\newcommand{\Perm}[1][n]{\sym\left(#1\right)}

\newcommand{\EndoC}{[\catC,\catC]}
\newcommand{\MonEndoC}{\catfont{Mon}[\catC,\catC]}
\newcommand{\SMEndoC}{\catfont{SymMon}[\catC,\catC]}

\newcommand{\Ban}{\catfont{Ban}}

\newcommand{\Cats}{\catfont{Cat}}

\newcommand{\catM}{\catfont{M}}

\newcommand{\Met}{\catfont{Met}}

\newcommand{\Free}[1]{\sfunfont{F}}

\newcommand{\typeI}{\typefont{I}}

\newcommand{\Shuff}{\mathrm{Sf}}
\newcommand{\VCat}{\mathcal{V}\text{-}\Cats}
\newcommand{\VCatSy}{\mathcal{V}\text{-}\Cats_{\mathsf{sym}}}
\newcommand{\VCatSe}{\mathcal{V}\text{-}\Cats_{\mathsf{sep}}}
\newcommand{\VCatSS}{\mathcal{V}\text{-}\Cats_{\mathsf{sym,sep}}}

\newcommand{\Syn}{\mathrm{Syn}}
\newcommand{\conc}{\mathbin{+\mkern-08mu+}}

\newcommand{\catfont}[1]{\mathsf{#1}}
\newcommand{\cop}{\catfont{op}}

\newcommand{\catC}{\catfont{C}}

\newcommand{\catA}{\catfont{A}}
\newcommand{\catB}{\catfont{B}}

\newcommand{\Set}{\catfont{Set}}

\newcommand{\Pos}{\catfont{Pos}}

\newcommand{\Subs}[2]{\catfont{Sub}_{}}

\newcommand{\funfont}[1]{#1}
\newcommand{\funF}{\funfont{F}}

\newcommand{\sfunfont}[1]{\mathrm{#1}}

\newcommand{\Id}{\sfunfont{Id}}

\newcommand{\Diag}{\mathscr{D}}

\newcommand\adjunct[2]{\xymatrix@=8ex{\ar@{}[r]|{\top}\ar@<1mm>@/^2mm/[r]^{{#2}}
& \ar@<1mm>@/^2mm/[l]^{{#1}}}}
\newcommand\adjunctop[2]{\xymatrix@=8ex{\ar@{}[r]|{\bot}\ar@<1mm>@/^2mm/[r]^{{#2}}
& \ar@<1mm>@/^2mm/[l]^{{#1}}}}
\newcommand\retract[2]{\xymatrix@=8ex{\ar@{}[r]|{}\ar@<1mm>@/^2mm/@{^{(}->}[r]^{{#2}}
& \ar@<1mm>@/^2mm/@{->>}[l]^{{#1}}}}

\newcommand{\sem}[1]{\llbracket #1 \rrbracket}

\newcommand{\app}{\mathrm{app}}

\newcommand{\comp}{\cdot}

\DeclareMathOperator{\id}{\mathsf{id}}
\DeclareMathOperator{\sw}{\gamma}
\DeclareMathOperator{\spl}{\mathsf{sp}}
\DeclareMathOperator{\sh}{\mathsf{sh}}
\DeclareMathOperator{\join}{\mathsf{jn}}
\DeclareMathOperator{\exch}{\mathsf{exch}}

\newcommand{\resp}{resp.\ }
\newcommand{\cf}{\textit{cf.\ }}
\newcommand{\ie}{\textit{i.e.\ }}
\newcommand{\eg}{\textit{e.g.\ }}
\newcommand{\wrt}{w.r.t.\ }

\newcommand{\iid}{i.i.d.\ }

\newcommand{\Nats}{\mathbb{N}}

\newcommand{\typefont}[1]{\mathbb{#1}}

\newcommand{\typeA}{\typefont{A}}
\newcommand{\typeB}{\typefont{B}}
\newcommand{\typeC}{\typefont{C}}

\newcommand{\rulename}[1]{(\mathrm{\mathbf{#1}})}
\newcommand{\vljud}{\rhd}

\newcommand{\prog}[1]{\mathtt{#1}}

\def\conf{MFPS 2023} 	
\def\myconf{mfps39}
\volume{3}			
\def\volu{3}			
\def\papno{4}			
\def\lastname{Dahlqvist, Neves} 
\def\runtitle{A complete $\V$-equational system for graded $\lambda$-calculus} 
\def\mydoi{12299}
\def\copynames{F. Dahlqvist, R. Neves}
\begin{document}
\begin{frontmatter}
  \title{A Complete $\V$-Equational System for Graded $\lambda$-Calculus} 						
  \author{Fredrik Dahlqvist\thanksref{a}\thanksref{fredemail}}	
   \author{Renato Neves\thanksref{b}\thanksref{renatoemail}}		
   \address[a]{Queen Mary University of London and University College London, United Kingdom}  							
   \thanks[fredemail]{Email: \href{mailto:f.dahlqvist@qmul.ac.uk} {\texttt{\normalshape
        f.dahlqvist@qmul.ac.uk}}} 
  \address[b]{University of Minho \& INESC-TEC, Portugal} 
  \thanks[renatoemail]{Email:  \href{mailto:nevrenato@di.uminho.pt} {\texttt{\normalshape
        nevrenato@di.uminho.pt}}}
\begin{abstract} 
        Modern programming frequently requires generalised notions of program
        equivalence based on a metric or a similar structure. Previous work
        addressed this challenge by introducing the notion of a $\V$-equation,
        \ie an equation labelled by an element of a quantale $\V$, which covers
        \textit{inter alia} (ultra-)metric, classical, and fuzzy (in)equations.
        It also introduced a $\V$-equational system for the \emph{linear}
        variant of $\lambda$-calculus where any given resource must be used
        exactly once.

        In this paper we drop the (often too strict) linearity constraint by
        adding graded modal types which allow multiple uses of a resource in a
        controlled manner. We show that such a control, whilst providing more
        expressivity to the programmer, also interacts more richly with
        $\V$-equations than the linear or Cartesian cases.  Our
        main result is the introduction of a sound and complete $\V$-equational
        system for a $\lambda$-calculus with graded modal types interpreted by
        what we call a \emph{Lipschitz exponential comonad}.  We also show how
        to build such comonads canonically via a universal construction, and
        use our results to derive graded metric equational systems (and
        corresponding models) for programs with timed and probabilistic
        behaviour.
\end{abstract}
\begin{keyword}
$\lambda$-calculus, graded modal type, quantitative equational theory,
        enriched category theory.
\end{keyword}
\end{frontmatter}

\section{Introduction}\label{sec:intro}

This paper tackles the challenge of reasoning about program equivalence in
computational paradigms with an intrinsic quantitative nature, such as timed
and probabilistic computation. This usually calls for notions of program
equivalence based on a  quantity (often a metric), \emph{in lieu} of the sharp,
binary ones relating classical programs. For example, instead of checking
whether two programs terminate \emph{exactly} at the same time one might be
more interested in checking whether they terminate with a small difference
between their execution times.  Similarly, on the probabilistic side, it makes
sense to consider that two Bayesian inference algorithms are equivalent if they
agree up to some small (total variation) error $\varepsilon$ when sampling from
the same target posterior distribution.  In order to reason in this
way,~\cite{dahlqvist22} introduced the notion of a $\V$-equation, \ie an
equation labelled by an element of a quantale $\V$, that serves as an abstract
notion of `quantitative equality'.  This covers, for example, (ultra-)metric and
fuzzy (in)equations, among others. Additionally \cite{dahlqvist22} presented a
$\V$-equational system for the \emph{linear} version of $\lambda$-calculus
which imposes that any given resource must be used exactly once.

The aim of this work is to overcome this linearity constraint whilst retaining
the ability to reason quantitatively about program equivalence. We do so by
adding \emph{graded modal
types}~\cite{girard1992bounded,gaboardi16,orchard2019quantitative} (a way of
permitting multiple uses of a given resource) to the aforementioned
$\V$-equational framework of linear $\lambda$-calculus~\cite{dahlqvist22}. The
result is a compromise between standard, non-linear $\lambda$-calculus which is
to some degree incompatible with quantitative reasoning (see the negative
results of~\cite[\S 6]{ugo2021}) and linear $\lambda$-calculus which can be
combined with quantitative reasoning~\cite{dahlqvist22} but is cumbersome for
many non-linear applications.

Let us illustrate this compromise with a simple example that involves metric
equations~\cite{mardare17} and timed computation~\cite{dahlqvist22}. Consider a
ground type $X$ and a signature $\{ \prog{wait_n} : X \to X \mid n \in \Nats
\}$ of wait calls -- intuitively, a term $\prog{wait_n}(x)$ reads as ``add a latency
of $n$ seconds to computation $x$". As discussed in~\cite{dahlqvist22}, a
series of metric equations arise naturally from this computational paradigm.
For example, 
\begin{align}
        \label{eq:ex}
        \lambda x. \, \prog{wait_1}(x) =_{1} \lambda x. \, \prog{wait_2}(x)
\end{align}
states that when fed the same argument these $\lambda$-terms yield computations
whose execution times differ by at most one second.  Now, as a useful principle
that underpins compositionality we would like that for all $\lambda$-terms $u$
the application function $v \mapsto u \, v$ satisfies the implication $v =_q w
\Rightarrow u \, v =_q u \, w$, \ie\ it is non-expansive w.r.t.  distances
between programs. This is impossible in the Cartesian setting, because $u$ may
contain multiple ocurrences of a variable (corresponding to multiple uses of a
given resource). Let $u$ for example be $\lambda f. \, \lambda y.  \, f \, (f
\, y)$.  Then $u \, (\lambda x. \, \prog{wait_1}(x))$ corresponds to an
execution time of two seconds and $u \, (\lambda x. \, \prog{wait_2}(x))$ to
four seconds, a two-second difference that violates the
implication for \eqref{eq:ex}.  The graded setting explored in this paper
serves as middleground between the linear and Cartesian cases: it
increases distances proportionally to the number of times a resource is usable
and at the same time forbids $u$ from using a resource more times than
stipulated. Specifically for the case just presented one can mark $\lambda x.
\, \prog{wait_1}(x)$ (resp.  $\lambda x.  \, \prog{wait_2}(x)$) to be usable
\emph{precisely twice}, via a `promotion construct' $!_2(-)$, and according
to our graded equational system deduce the metric equation,
\[!_2
\big ( \lambda x.  \, \prog{wait_1}(x) \big ) =_{1+1 } \, !_2 \big (\lambda x. \,
\prog{wait_2}(x) \big )
\]
We then use the graded typing system to ensure $u$ uses the received argument
precisely twice. We will see that this ensures the non-expansiveness of the
application function -- actually of the more general case $(u,v) \mapsto u \,
v$ -- amongst other benefits.

\noindent
\textbf{Contributions and outline.}
We present a sound and complete $\V$-equational system for a graded
$\lambda$-calculus. The corresponding interpretation is based on symmetric
monoidal closed categories enriched over `generalised metric spaces' and
equipped with a \emph{Lipschitz exponential comonad}, a natural extension of
the concept of graded exponential comonad~\cite{gaboardi16,katsumata18} to the
setting of $\V$-equations.  Furthermore, we show how to canonically build
Lipschitz exponential comonads over symmetric monoidal closed categories that
satisfy mild conditions. The construction is inspired by~\cite{mellies09}, and
based on the notion of a cofree graded commutative comonoid together with a
certain kind of enriched limit.

\S\ref{sec:calculus} introduces a graded $\lambda$-calculus and
an equational system that characterises term equivalence. This
calculus fundamentally differs from previous
ones~\cite{brunel14,gaboardi16,orchard2019quantitative} in that the
substitution rule in its standard format is derivable -- this is key to our
completeness result. \S\ref{sec:calculus} also presents an interpretation of
the calculus via symmetric monoidal closed (a.k.a.\ autonomous) categories
together with graded exponential comonads~\cite{gaboardi16,katsumata18}. It
then proves soundness of the aforementioned equational system \wrt this
interpretation.  \S\ref{sec:quant_eq} extends \S\ref{sec:calculus} to the
$\V$-equational setting. Specifically, it equips our graded $\lambda$-calculus
with a $\V$-equational system and shows how to interpret it via autonomous
categories enriched over generalised metric spaces together with Lipschitz
exponential comonads. It also shows that the $\V$-equational system is sound
and complete \wrt this interpretation (Theorem \ref{theo:compl}). This result is highly generic and covers metric equations, classical (in)equations and ultra-metric and fuzzy variants. To the best of our knowledge this
completeness result even for the basic case of classical equations is new.
\S\ref{sec:canonical} details the aforementioned canonical construction of Lipschitz
exponential comonads and \S\ref{sec:examples} uses it as basis to provide
metric higher-order models of both timed and probabilistic computation. In the
former case the model that we canonically obtain is based on the category of
metric spaces and non-expansive maps with the underlying Lipschitz comonad
being that of dilations~\cite{katsumata18}. In the latter case the model is
based on the category of Banach spaces and short linear maps with the
underlying Lipschitz comonad arising from a process of symmetrisation
well-known in linear algebra~\cite{bourbaki98,comon08}.  We assume basic knowledge of (enriched)
category theory.

\noindent
\textbf{Related work.}
The need for quantitative notions of program equivalence has been 
explored in several concrete computational paradigms. This is the case for example
of~\cite{reed10}, \cite{hung19}, and~\cite{crubille15,crubille17} which
introduce metric reasoning mechanisms for differential privacy, quantum, and
probabilistic computation respectively. Other works take a more general
perspective. For example on the side of universal algebra there has been
great progress on the closely related topic of quantitative algebra, with 
focus typically on metric equations and inequations, see for
example~\cite{mardare16,mardare17,rosicky21,adamek21}. In fact, one case with a
particularly interesting connection to ours is~\cite{dagnino22}: it explores a
notion of quantitative equality with graded modalities and studies a
corresponding \emph{algebraic semantics} via Lawvere's doctrines. Our target
is, however, $\lambda$-calculus. This sets us apart from these approaches, and in this
regard positions us closer to the quantitative approaches targetting
$\lambda$-calculi such as~\cite{gavazzo18}
and~\cite{gavazzo23} which use the notion of a quantale to introduce
quantitative counterparts of applicative (bi)similarity and rewriting systems
respectively. Another example is~\cite{pistone21} which studies quantitative
semantics of simply-typed $\lambda$-calculi based on a generalisation of
logical relations.

\section{A graded $\lambda$-calculus and its interpretation} 
\label{sec:calculus}

\subsection{The calculus}

We start by presenting our graded $\lambda$-calculus. In a nutshell, it is a
graded extension of the linear-non-linear $\lambda$-calculus
in~\cite{benton92,benton1994mixed} and can be seen as a term assignment system
for a graded version of \emph{intuitionistic linear logic}. Aside from the use
of grades, the main difference with~\cite{benton92,benton1994mixed} is the use
of a shuffling mechanism~\cite{shulman19} that allows to refer to a
$\lambda$-term's denotation unambiguously (more details below).

\noindent \textbf{Types.} As usual with graded modal
types~\cite{girard1992bounded,gaboardi16,orchard2019quantitative}, we
fix a semiring $\mathcal{R} = (R,0,1,+, \, \cdot)$ of `resource quantities'.
We then fix a set $G$ of ground types and consider the following
grammar of types:
\[
  \typeA ::=  X \mid \typeI \mid 
  \typeA \otimes \typeA \mid \typeA \multimap \typeA \mid\ \g{r} \typeA 
  \hspace{2.5cm} (X \in G, r \in \mathcal{R}).
\]
Elements of $R$ will be called \emph{grades}.  The grade $r$ associated with a
modal type $\g{r} \typeA$ intuitively represents \emph{how much} of a resource
we possess. For example, in the case of $\mathcal{R}$ being the semiring of
natural numbers $r$ may be regarded as the number of times a resource can be
used before depletion.

\noindent\textbf{Contexts and shuffles.} We use Greek uppercase letters
$\Gamma, \Delta, E, \dots$ to denote typing contexts, \ie lists of typed
variables $x_1 : \typeA_1, \dots, x_n : \typeA_n$ such that each $x_i$ occurs
at most once. As already mentioned, we will also use the notion of a shuffle: a
permutation of typed variables in a context sequence $\Gamma_1,\dots,\Gamma_n$
such that for all $i \leq n$ the relative order of the variables in $\Gamma_i$
is preserved~\cite{shulman19}.  For example, if $\Gamma_1 = x : \typeA, y :
\typeB$ and $\Gamma_2 = z : \typeC$ then $z : \typeC, x : \typeA, y : \typeB$
is a shuffle but $y : \typeB, x : \typeA, z : \typeC$ is \emph{not}, because we
changed the order in which $x$ and $y$ appear in $\Gamma_1$. We denote by
$\Shuff(\Gamma_1;\dots;\Gamma_n)$ the set of shuffles on
$\Gamma_1,\dots,\Gamma_n$. Shuffles will be used to build a graded
$\lambda$-calculus where the exchange rule is admissible and at the same time
each judgement $\Gamma \vljud v : \typeA$ has a unique derivation
(Theorem~\ref{thm:unique_typing}). This will allow us to refer to a judgement's
denotation $\sem{\Gamma \vljud v : \typeA}$ unambiguously.

\noindent\textbf{Terms.} Fix a set $\Sigma$ of sorted operation symbols $f:
\typeA_1,\dots,\typeA_n \to \typeA$ with $n \geq 1$. The term formation rules
of the graded calculus are listed in Figure~\ref{fig:lang}. By convention all
contexts involved in the premisses of any of the listed rules are mutually
disjoint. This entails for instance that in $\rulename{\otimes_e}$ neither $x$
nor $y$ can occur in $\Gamma$ and analogously for $\rulename{!_{{n+m}}}$. The
rules above the dotted line are standard and in correspondence to the natural
deduction rules of exponential-free intuitionistic linear logic; we omit here
their explanation.  As for the others, the promotion rule $\rulename{!_i}$
allows the use of a term `$r$-times' by intuitively binding all variables $x_i
: \, \g{s_i} \typeA_i$ in its context to terms $v_i$ whose type $\g{r\cdot s_i}
\typeA_i$ is graded by the `$r$-multiple' of $s_i$.  The dereliction rule
$\rulename{!_e}$ connects the modal typing system to the linear one, in
particular it makes explicit that terms with linear types must be used exactly
once. This is essential \eg for using terms whose type is linear multiples
times. Take for example the semiring of natural numbers and a sorted operation
symbol $f : \typeA \to \typeA$. A call to $f$ that is usable precisely
`$r$-times' is given by the judgement $y :\, \g{r} \typeA \vljud \,
\prog{pr}_{(r,[1])} \, y \, \prog{fr} \, x. \, f(\prog{dr} \, x) : \, \g{r}
\typeA$. Finally rules $\rulename{!_{{0}}}$ and $\rulename{!_{{n+m}}}$
correspond respectively to graded versions of weakening and contraction. They
can be seen intuitively as discard and copy operations where in the latter case
variables $x$ and $y$ are bound to the object $v$ being copied.

\begin{remark}
        When we instantiate $\mathcal{R}$ to the trivial semiring $(\{\infty\},
        \infty, \infty, +, \, \cdot)$, the rules in Figure~\ref{fig:lang} are
        the ones presented in~\cite{benton92} modulo the shuffling mechanism.
\end{remark}

\begin{figure*}[h!]
    \centering
    \scalebox{0.93}{
    \begin{tabular}{llllc}
    & & & & \\
    \multicolumn{4}{l}{
     \infer[\rulename{ax}]{E
        \vljud f(v_1,\dots,v_n) : \typeA}
      {\Gamma_i \vljud v_i : \typeA_i
      \quad f : \typeA_1,\dots,\typeA_n \to \typeA \in \Sigma
      \quad E \in \Shuff(\Gamma_1;\dots;\Gamma_n) }
      }
      &
      \infer[\rulename{hp}]{x : \typeA \vljud x : \typeA}{}
      \\
      & & & & \\
      \infer[\rulename{\typeI_i}]{- \vljud \ast : \typeI}{}
      &
      \multicolumn{4}{r}{
      \infer[\rulename{\typeI_e}]{E  \vljud v \ \prog{to}\ \ast.\ w
      : \typeA}
      {\Gamma \vljud v : \typeI \quad \Delta 
        \vljud w : \typeA \quad E \in \Shuff(\Gamma;\Delta)}
      }
      \\
     & & & & \\
      \multicolumn{3}{l}{
      \infer[\rulename{\otimes_i}]{E \vljud v \otimes w : \typeA \otimes
        \typeB}{\Gamma \vljud v : \typeA \quad \Delta \vljud w : \typeB
      \quad E \in \Shuff(\Gamma;\Delta)}
      }
      &
      \multicolumn{2}{r}{
      \infer[\rulename{\otimes_e}]{E \vljud \prog{pm}\ v\ \prog{to}\
      x \otimes y.\ w : \typeC}
      {\Gamma \vljud v : \typeA \otimes \typeB
      \quad \Delta , x : \typeA, y : \typeB \vljud w : \typeC
      \quad E \in \Shuff(\Gamma;\Delta)}
      }
      \\
      & & & & \\
      \multicolumn{3}{l}{
      \infer[\rulename{\multimap_i}]{\Gamma \vljud \lambda x : \typeA . \,
      v : \typeA \multimap \typeB}
      {\Gamma, x : \typeA \vljud v : \typeB}
 
      }
      &
      \multicolumn{2}{r}{
      \infer[\rulename{\multimap_e}]{E \vljud v \, w : \typeB}
      {\Gamma \vljud  v : \typeA \multimap \typeB \quad
      \Delta \vljud  w : \typeA \quad E \in \Shuff(\Gamma;\Delta)}      
     }
     \\ 
     \multicolumn{5}{c}{
        \dotfill 
     }
     \\
     & & & & \\
     \multicolumn{4}{c}{
             \infer[\rulename{!_i}]{
                     E \vljud \prog{pr}_{(r,[s_1,\dots,s_n])}\ v_1,\dots,v_n\ \prog{fr}\
                     x_1,\dots,x_n .\ u:\ \g{r} \typeA
             }{
                     \Gamma_i \vljud v_i :\ \g{r \cdot s_i} \typeA_i \quad
                     x_1 :\ \g{s_1}\typeA_1,\dots,x_n :\ \g{s_n} \typeA_n \vljud u : \typeA
                     \quad E \in \Shuff(\Gamma_1;\dots;\Gamma_n)  
             }
     }
     &
         \infer[\rulename{!_e}]{
               \Gamma \vljud \prog{dr}\, v : \typeA
       }{
               \Gamma \vljud v :\ \g{1} \typeA
       }

     \\
     &&&&
     \\
     \multicolumn{2}{l}{
             \infer[\rulename{{!_{0}}}]{
                     E \vljud \prog{ds}\, v .\ u : \typeB
             }
             {
                     \Gamma \vljud v :\ \g{0} \typeA \quad
                     \Delta \vljud u : \typeB \quad
                     E \in \Shuff(\Gamma; \Delta)
             }
     }
     &
     \multicolumn{3}{r}{
             \infer[\rulename{{!_{n+m}}}]{
                     E \vljud \prog{cp}_{(n,m)}\ v\ \prog{to}\ x,y.\ u : \typeB
             }{
                     \Gamma \vljud v :\ \g{n + m} \typeA \quad
                     \Delta, x :\ \g{n} \typeA, y :\ \g{m} \typeA \vljud
                     u : \typeB \quad E \in \Shuff(\Gamma;\Delta)
             }
     }
    \end{tabular}
        }
  \caption{Term formation rules of graded $\lambda$-calculus.}
  \label{fig:lang}
\end{figure*}

\noindent\textbf{Properties.} Our calculus has several desirable properties
(Theorem~\ref{thm:unique_typing} and Lemma~\ref{lem:fexch_subst}),
including the aforementioned fact that all judgements have a  unique derivation.
We start by presenting auxiliary notations.  Given a context $\Gamma$ we will
use $te(\Gamma)$ to denote  $\Gamma$ with all types erased.  Additionally, for
contexts $\Gamma$ and $\Gamma'$ we will use notation $\Gamma\perm \Gamma'$ to
state that $\Gamma$ is a permutation of $\Gamma'$. We will also use an
analogous notation for non-repetitive lists of untyped variables $te(\Gamma)$.
We will often abbreviate a judgement $\Gamma \vljud v : \typeA$ into $\Gamma
\vljud v$ or even just $v$ if no ambiguities arise.  Finally, we will often
denote a list of terms $v_1,\dots,v_n$ simply by $\vec{v}$ and analogously for
lists of variables.
\begin{proposition}
        \label{prop:shuff}
        Let us consider two lists of contexts $\Gamma_1,\dots, \Gamma_n$ and
        $\Gamma'_1,\dots,\Gamma'_n$, contexts $E$ and $ E'$, and suppose that
        $E \in \Shuff(\Gamma_1;\dots;\Gamma_n)$, $E' \in
        \Shuff(\Gamma'_1;\dots;\Gamma'_n)$. Then the following clauses hold: 
        \begin{enumerate}
                \item if $te(\Gamma_i)\perm te(\Gamma'_i)$ for all $i \leq n$
                        then $te(E)\perm te(E')$;
                \item if $\Gamma_i\perm\Gamma'_i$ for all $i \leq n$
                        then $E\perm E'$;
                \item if $E\perm E'$ and $te(\Gamma_i)\perm te(\Gamma'_i)$ 
                        for
                        some $i \leq n$ then $\Gamma_i\perm \Gamma'_i$;
                \item if $E = E'$ and $te(\Gamma_i)\perm te(\Gamma'_i)$  for some
                        $i \leq n$ then $\Gamma_i = \Gamma'_i$.
        \end{enumerate}

\end{proposition}

\begin{theorem}
        \label{thm:unique_typing}
        Graded $\lambda$-calculus has the following properties:
        \begin{enumerate}
                \item for all judgements $\Gamma \vljud v$ and $\Gamma'
                        \vljud v$ we have $te(\Gamma)\perm
                        te(\Gamma')$;
                \item additionally if $\Gamma \vljud v : \typeA$, $\Gamma' \vljud v:
                        \typeA'$, and $\Gamma\perm \Gamma'$ then
                        $\typeA$ must be equal to $\typeA'$;
                \item all judgements $\Gamma \vljud v : \typeA$ have a 
                        unique
                        derivation.
        \end{enumerate}
 
\end{theorem}
\begin{proof}
        The first clause follows straightforwardly from induction over the
        derivation system (Figure~\ref{fig:lang}) and the first clause of
        Proposition~\ref{prop:shuff}. The second clause follows from induction
        over the derivation system, the first clause, the second and third
        clauses of Proposition~\ref{prop:shuff}, the grade annotations in term
        constructs, and the type annotation in the $\lambda$-construct.  The
        third clause follows from induction over the derivation system, the
        second clause, the shuffling mechanism, and the fourth clause of
        Proposition~\ref{prop:shuff}.
\end{proof}

Substitution is defined in the expected way and as usual uses
$\alpha$-equivalence to avoid capturing free variables. In our setting such
captures arise from the rules $\rulename{\multimap_i}$, $\rulename{\otimes_e}$,
$\rulename{!_i}$, and $\rulename{!_{{n+m}}}$.

\begin{lemma}[Exchange and Substitution]
  \label{lem:fexch_subst}
  For every judgement
  $\Gamma, x : \typeA, y : \typeB, \Delta \vljud v : \typeC$ we can derive
  $\Gamma, y : \typeB, x : \typeA, \Delta \vljud v : \typeC$.
  For all judgements
  $\Gamma, x : \typeA \vljud v : \typeB$ and
  $\Delta \vljud w : \typeA$ we can derive
  $\Gamma,\Delta \vljud v[w/x] : \typeB$.
\end{lemma}
\begin{proof}
        As usual the exchange property follows from induction over the
        derivation system in Figure~\ref{fig:lang}. The substitution property
        follows from the exchange property, the fact that $x$ occurs at most
        once in the term $v$, and from induction over the judgement derivation
        $\Gamma, x : \typeA \vljud v : \typeB$.
\end{proof}

The substitution property proved in Lemma~\ref{lem:fexch_subst} generalises to
iterated substitution. More specifically, given $\Gamma, x_1 : \typeA_1, \dots,
x_n : \typeA_n \vljud v : \typeB$  and $\Delta_i \vljud w_i : \typeA_i$ $(i
\leq n)$ with all contexts involved pairwise disjoint one easily derives
$\Gamma, \Delta_1, \dots, \Delta_n \vljud v[w_1/x_1] \dots [w_n/x_n] : \typeB$.
Additionally it is straightforward to prove that, by virtue of all contexts
being pairwise disjoint, the order in which the sequence of substitutions
occurs is irrelevant. For this reason we will often abbreviate $v[w_1/x_1]
\dots [w_n/x_n]$ simply to $v[\vec{w}/\vec{x}]$ or $v[w_1/x_1,\dots,w_n/x_n]$.

\begin{remark}
The promotion rule $\rulename{!_i}$ of our graded calculus differs from the
promotion rule of previous calculi with graded
modalities~\cite{brunel14,gaboardi16}. Let us explain this distinction and
justify it. Let $\vec{s}$ denote a list of grades $s_1,\dots,s_n$ and $r\cdot
\vec{s}$ denote the list of grades $r \cdot s_1, \dots, r\cdot s_n$. If we
write $!_{\vec{s}}\ \Gamma$ to say that the type of every variable $x_i$ in
$\Gamma$ is of the form $!_{s_i}\, \typeA_i$, then for every judgement
$!_{\vec{s}}\ \Gamma \vljud v : \typeA$ with $te(\Gamma) = x_1,\dots,x_n$ we
can derive $!_{r\cdot \vec{s}}\ \Gamma \vljud \prog{pr}_{(r,\vec{s})}\ \vec{x}
\ \prog{fr}\ \vec{y}.\, v[\vec{y}/\vec{x}] :\, !_r\, \typeA$ -- we abbreviate
the latter term simply to $\g{r} v$. The following rule is then 
admissible in our calculus:
\[
        \infer[]{!_{r \cdot
                \vec{s}}\ \Gamma \vljud \> \g{r} v : \, !_r\, \typeA}{!_{\vec{s}}\ 
        \Gamma \vljud v : \typeA}
\]
A rule with the same structural format is added \emph{natively} to the calculi
in~\cite{brunel14,gaboardi16} and is the counterpart to our promotion rule
$\rulename{!_i}$. The former however breaks the substitution property stated in
Lemma~\ref{lem:fexch_subst} (details available in~\cite[page 10]{benton92}).
This would hinder the development of our equational system and associated
completeness result and justifies the slightly more complicated rule
$\rulename{!_i}$. 
\end{remark}

\noindent\textbf{Equational system.} Figure~\ref{fig:eqs} presents the equational schema of graded
$\lambda$-calculus. As usual, we omit the typing information of the
equations-in-context listed in Figure~\ref{fig:eqs} which can be recovered
uniquely up to permutations. The symbols $(:)$ and $(\conc)$ denote usual
operations on lists namely cons and concatenation. Note as well the division of
the equational schema into different sections referring to specific categorical
machinery.  This is to attach a semantic intuition to the equations and to
foreshadow the categorical structures that will be used later on to interpret
graded $\lambda$-calculus.  The equations concerning the monoidal structure and
the closed structure were already discussed elsewhere~(\eg
\cite{benton92,dahlqvist22}).  The equations concerning commuting conversions
enforce the fact that certain expressions differing in scope such as
$(\prog{ds} \, v. \, u) \otimes w$ and $\prog{ds} \, v. \, (u \otimes w)$ are
intended to have the same meaning.

Next, in the axiomatisation of the comonadic structure, the first and second
equations  are respectively $\beta$ and $\eta$ equations and embody the counit
laws associated to the underlying graded comonad. The third equation states
that the inner promotion (on the left-hand side) can be pushed-forward to $w$
but with the factor $r_1$ discarded as a result from not being bound to
variable $a$ anymore. This equation embodies the associativity law  of the
underlying graded comonad.  Observe that for these three equations to be
well-defined the reduct $(R,1,\, \cdot \,)$ in the semiring $\mathcal{R}$ needs
to be a monoid (which we assumed previously).  The fourth equation tells that
the order in which terms $\vec{v}$  appear in a promotion
$\prog{pr}_{(r,\vec{s})} \, \vec{v} \, \prog{fr} \, \vec{x} .\,  u $ is
irrelevant, which fact embodies the symmetry of the graded comonad.

The discard (\ie weakening) and copy (\ie contraction) operations suggest a
(graded) commutative comonoidal structure, which is reflected in the four
corresponding equations in Figure~\ref{fig:eqs}. This time, these equations
force the reduct $(R,0,\, + \,)$ in the semiring $\mathcal{R}$ to be a
commutative monoid (which indeed we also assumed previously).  In the
axiomatisation of the interaction between the underlying comonoid and comonad,
the first two equations can be seen as a mechanism for shifting term complexity
between the discard and promotion expressions (this is noticeable by looking at
the grade annotations in the promotions, when present). They may equally well
be regarded respectively as $\beta$ and $\eta$-equations whose corresponding
reduction simplifies the promotion expression. Semantically they reflect the
naturality of the discard operation, that the latter is a graded version of a
coalgebra morphism, and that the comonad's comultiplication is a comonoid
morphism (we formally detail this later on).  Note as well that these equations
force $0$ to be an absorbing element of the monoid operation $(\cdot)$ in the
semiring $\mathcal{R}$ (which indeed we assumed previously). The last two
equations follow a reasoning analogous to the previous two, and force $(\cdot)$
to distribute over $(+)$ both on the left and the right (which we also
assumed). The equations described thus entail that $\mathcal{R}$ has a semiring
structure as previously postulated.  

\begin{remark}
This equational schema is a graded generalisation of the one
presented in~\cite{benton92}. In fact, for the particular case of the singleton
semiring $\mathcal{R} = \{\{\infty\},\infty,\infty,+,\, \cdot\}$ our equations
collapse to those in~\cite{benton92} except for the equation about the
comonad's symmetry which is absent from~\emph{op.\ cit.}
\end{remark}

\begin{figure*}[h!]
\captionsetup{width=\linewidth}
        \centering
        \scalebox{0.92}{
	\begin{tabular}{| c | c |}
		\hline
		Monoidal structure
		&
		Closed structure
                \\
		\hline
		\begin{tabular}{r c l}
	           $\prog{pm}\ v \otimes w\ \prog{to}\ x \otimes y.\ u$ &
	           $=$ & $u[v/x,w/y]$ 
                   \\[3pt]
	           $\prog{pm}\ v\ \prog{to}\ x \otimes y.\
	           u[x \otimes y / z]$ &
	           $=$ & $u[v/z]$ 
                   \\[3pt]
	           $\ast\ \prog{to}\ \ast.\ v$ & $=$ & $v$ \\
	           $v\ \prog{to}\ \ast.\ w[\ast/ z]$ & $=$ & $w[v/z]$
                   \\[2pt]
	 \end{tabular}
  	 &
	 \begin{tabular}{r c l}
                      $(\lambda x : \typeA .\ v)\ w$ &$=$& $v[w/x]$ 
                      \\[3pt]
                      $\lambda x : \typeA .\ v\ x$ &$=$& $v$
         \end{tabular}
	 \\
         \hline
         \multicolumn{2}{|c|}{Symmetric comonadic structure} 
         \\
         \hline
         \multicolumn{2}{|c|}{
         \begin{tabular}{r c l}
                $\prog{dr} \, \,
                \prog{pr}_{(1,\vec{s})}\ \vec{v}\ \prog{fr}\ \vec{x}.\ u
                $
                &$=$& $u[\vec{v}/\vec{x}]$ 
                \\[3pt]
                 $\prog{pr}_{(r,[1])}\ z\ \prog{fr}\ x.\ \prog{dr}\, x$ & $=$ & $z$
                \\[3pt] 
                $\prog{pr}_{(r_1,r_2 : \vec{r})}\  (\prog{pr}_{(r_1 \cdot r_2,\vec{s})}\
                \vec{x}\ \prog{fr}\ \vec{y}.\ v  ), \vec{z}\ \prog{fr}\ a,\vec{a}.\ w$
                & $=$ &
                $\prog{pr}_{(r_1, (r_2 \cdot \vec{s}) \conc \vec{r})}\ \vec{x},\vec{z}\
                \prog{fr}\ \vec{c},\vec{a}.\ w [\prog{pr}_{(r_2,\vec{s})}\ \vec{c}\
                \prog{fr}\ \vec{y} .\ v /a ]$ 
                \\[3pt]
                $\prog{pr}_{(r,\vec{s_1}\conc [r_1,r_2]
                \conc \vec{s_2})}\ \vec{v}_1,w_1,w_2,
                \vec{v}_2\ \prog{fr}\ \vec{x}_1,y_1,y_2,\vec{x}_2.\ u$
                &$=$& $\prog{pr}_{(r,\vec{s_1}\conc [r_2,r_1] \conc 
                \vec{s_2})}\ \vec{v}_1,w_2,w_1,
                \vec{v}_2\ \prog{fr}\ \vec{x}_1,y_2,y_1,\vec{x}_2.\ u$
                \\[3pt]
        \end{tabular}
        }\\[3pt] 
        \hline
        \multicolumn{2}{|c|}{Commutative comonoid structure} 
        \\
        \hline
        \multicolumn{2}{|c|}{
         \begin{tabular}{r c l}
                 $\prog{cp}_{(0,n)}\, v\, \prog{to}\, x,y.\, \prog{ds}\, x.\, u$
                &$=$& $u[v/y]$ 
                \\[3pt]
                $\prog{cp}_{(n,0)}\, v\, \prog{to}\, x,y.\, \prog{ds}\, y.\, u$ & $=$ & 
                $u[v/x]$
                \\[3pt] 
                $\prog{cp}_{(n+m,o)}\, v\, \prog{to}\, x,y.\,
                        \prog{cp}_{(n,m)}\, x \, \prog{to} \, a,b.\, u$
                & $=$ &
                $\prog{cp}_{(n,m+o)} \, v \, \prog{to} \, a,c .\, 
                        \prog{cp}_{(m,o)} \, c \, \prog{to}\, b,y.\, u$ 
                \\[3pt]
                $\prog{cp}_{(n,m)}\, v\, \prog{to}\, x,y.\, u$
                &$=$& $\prog{cp}_{(m,n)}\, v\, \prog{to}\, y,x.\, u$
                \\[3pt]
        \end{tabular}
        }\\[3pt] 
        \hline
        \multicolumn{2}{|c|}{Interaction between comonoid and comonad} 
        \\
        \hline
        \multicolumn{2}{|c|}{
         \begin{tabular}{r c l}
                 $\prog{ds}\, \, \prog{pr}_{(0,\vec{s})}\, \vec{v} \, 
                         \prog{fr} \, \vec{x}.\ w .\ u$
                &$=$& $\prog{ds}\, v_1 .\, \dots \, \prog{ds}\, v_n.\, u$ 
                \\[3pt]
                 $\prog{pr}_{(r,0:\vec{s})}\, v, \vec{v}\, 
                 \prog{fr}\, x,\vec{x}.\, \prog{ds} \, x.\, u$ & $=$ & 
                 $\prog{ds}\, v.\, \prog{pr}_{(r,\vec{s})}\, \vec{v}\, \prog{fr} \,
                 \vec{x}. \, u$
                \\[3pt] 
                $\prog{cp}_{(n,m)} \, \prog{pr}_{(n+m,[s_1,\ldots,s_k])} \, \vec{v} \,
                \prog{fr} \, \vec{x}.\, w  \, \prog{to} \, y,z.\, u$
                & $=$ &
                $\prog{cp}_{(n \cdot s_1, m \cdot s_1)} \, v_1 \, \prog{to} \, a_1,b_1. \,
                \dots \, \prog{cp}_{(n \cdot s_k, m \cdot s_k)} \, v_k \, \prog{to} 
                \, a_k,b_k. \, $ 
                \\
                & & \hspace{0.2cm} 
                $u[\prog{pr}_{(n,[s_1,\ldots,s_k])} \, \vec{a} \, \prog{fr}\, \vec{x} \, 
                .\, w /y,\prog{pr}_{(m,[s_1,\ldots,s_k])} \, \vec{b} \, \prog{fr}\, \vec{x} \, 
                .\, w /z]$ 
                \\[3pt]
                $\prog{pr}_{(r,(n+m):\vec{s})} \, v,\vec{v} \, \prog{fr} \, z,\vec{z}. \,
                \prog{cp}_{(n,m)} \, z \, \prog{to} \, x,y. \, u$
                &$=$& $\prog{cp}_{(r\cdot n, r \cdot m)} \, v \, \prog{to} \, a,b. \,
                        \prog{pr}_{(r,n:m:\vec{s})} a,b, \vec{v} \, \prog{fr} \, x,y,\vec{z}
                        . \, u$
                \\[3pt]
        \end{tabular}
        }
        \\[3pt]
                \hline
                \multicolumn{2}{|c|}{Commuting conversions} 
                \\
                \hline
                \multicolumn{2}{|c|}{
                \begin{tabular}{r c l}
                        $u[v\ \prog{to} \ast.\ w/z]$ &$=$& $v\ \prog{to}\ \ast.\ u[w/z]$ 
                        \\[3pt]
	                $u[\prog{pm}\ v\ \prog{to}\ x \otimes y.\ w/z]$
	                &$=$& $\prog{pm}\ v\ \prog{to}\ x \otimes y.\ u[w/z]$
                        \\[2pt]
                        $u[\prog{ds}\ v.\, w/z]$ &$=$& $\prog{ds}\ v.\ u[w/z]$ 
                        \\[3pt]
                        $u[\prog{cp}_{(n,m)}\ v\ \prog{to}\ x,y.\ w/z]$
                        &$=$& $\prog{cp}_{(n,m)}\ v\ \prog{to}\ x, y.\ u[w/z]$
                        \\[2pt]
                \end{tabular}
                }\\
	  	\hline
	\end{tabular}
        }
	\caption{Equational schema of graded $\lambda$-calculus.}
  \label{fig:eqs}
\end{figure*}

\subsection{The interpretation}

In this subsection we present an interpretation of the graded calculus detailed
above. The interpretation uses the categorical machinery suggested
in~\cite{gaboardi16,katsumata18,orchard2019quantitative} to interpret previous
graded calculi.  We also prove that the equational schema in
Figure~\ref{fig:eqs} is sound w.r.t.\ this interpretation.

We start by recalling preliminary categorical notions and some conventions
concerning symmetric monoidal closed (\ie autonomous) categories.  Given one
such category $\catC$ and for a list of $\catC$-objects $X_1,\dots,X_n$ we
write $X_1 \otimes \dots \otimes X_n$ for the $n$-tensor $(\dots (X_1 \otimes
X_2) \otimes \dots ) \otimes X_n$ and similarly for morphisms. For
all $\catC$-objects $X,Y,Z$, $\sw : X \otimes Y \to Y \otimes X$ denotes the
symmetry morphism, $\lambda : I \otimes X \to X$ the left unitor, $\app : (X
\multimap Y) \otimes X \to Y$ the application morphism, and $\alpha : X \otimes
(Y \otimes Z) \to (X \otimes Y) \otimes Z$ the left associator.  For
all $\catC$-morphisms $f : X \otimes Y \to Z$ we denote the corresponding
curried version by $\overline{f} : X \to (Y \multimap Z)$. We will frequently
omit subscripts in natural transformations. For a monoidal functor $\funF
: \catC \to \catC$  we denote by $\phi : I \to \funF I$ and $\phi_{X,Y} : \funF
X \otimes \funF Y \to \funF (X \otimes Y)$ the corresponding monoidal
operations.  Similarly given $\catC$-objects $X_1,\dots,X_n$ we denote by
$\phi_{X_1,\dots,X_n} : \funF X_1 \otimes \dots \otimes \funF X_n \to \funF(X_1
\otimes \dots \otimes X_n)$ the morphism defined recursively on the size of
$n$ by:
\[
        \phi_{-} = \phi \hspace{2.7cm}
        \phi_{X} = \id \hspace{2.7cm}
        \phi_{X_1,\dots,X_n,X_{n+1}} = 
        \phi_{{(X_1 \otimes \dots \otimes X_n)},X_{n+1}} \comp
        (\phi_{X_1, \dots, X_n} \otimes \id).
\]
In the presence of several monoidal functors $F,G$, we denote their respective
monoidal operations by $\phi^F, \phi^G$.

We now set the ground for the notion of a graded exponential comonad,
explored for example in~\cite{gaboardi16,katsumata18,orchard2019quantitative} and
standardly used for interpreting graded modal types. Note first that a 
semiring $\mathcal{R}=(R,0,1,+,\cdot)$ has two (interacting) monoidal
structures: $(R,0,+)$ (which is commutative) and $(R,1,\cdot)$ (which need not
be).  The category $\EndoC$ of endofunctors and natural transformations also
has two monoidal structures, specifically $(\EndoC,I,\otimes)$ (where $I$
designates to constant functor to the unit) and $(\EndoC,\Id,\circ)$. The
category $\MonEndoC$ (\resp $\SMEndoC$) of \emph{monoidal} (\resp
\emph{symmetric monoidal}) endofunctors and \emph{monoidal} natural
transformations inherits these two monoidal structures from $\EndoC$.  The
semantics of our graded $\lambda$-calculus relies on a `representation' of
$\mathcal{R}$ in $\catC$ using these two structures, as detailed below.
\begin{definition}
        \label{defn:graded_com}
        An $\mathcal{R}$-\emph{graded comonad} over a (not necessarily
        monoidal) category  $\catC$ is an oplax monoidal functor $D:
        (R,1,\cdot)\to (\EndoC,\Id,\circ)$. Similarly,  an $\mathcal{R}$-graded
        \emph{monoidal} comonad is an oplax monoidal functor $D: (R,1,\cdot)\to
        (\MonEndoC,\Id,\circ)$,  and an $\mathcal{R}$-graded \emph{symmetric
        monoidal} comonad is an oplax monoidal functor $D: (R,1,\cdot)\to
        (\SMEndoC,\Id, \circ)$.  Concretely,  an $\mathcal{R}$-graded comonad
        is a triple $(D_{(-)} : R \to [\catC,\catC], \epsilon : D_1 \to \Id,
        \delta^{m,n} : D_{m \cdot n} \to D_m D_n)$ that makes the following
        diagrams commute
        \begin{align}\label{eq:gradedcomonad}
                \xymatrix@C=50pt@R=15pt{
                        D_s \ar@{=}[dr] \ar[r]^{\delta^{s,1}} \ar[d]_{\delta^{1,s}} 
                        & D_s D_1 \ar[d]^{D_s \epsilon}
                        \\
                        D_1 D_s \ar[r]_{\epsilon_{D_s}} & D_s
                }
                \hspace{2cm}
                \xymatrix@C=50pt@R=15pt{
                        D_{s_1 \cdot s_2 \cdot s_3} 
                        \ar[r]^{\delta^{s_1, s_2 \cdot s_3}} 
                        \ar[d]_{\delta^{s_1 \cdot s_2, s_3}} &  
                        D_{s_1} D_{s_2 \cdot s_3}
                        \ar[d]^{D_{s_1} \delta^{s_2,s_3}}
                        \\
                        D_{s_1 \cdot s_2} D_{s_3} 
                        \ar[r]_{{\delta^{s_1,s_2}}_{D_{s_3}}}
                        & 
                        D_{s_1} D_{s_2} D_{s_3} 
                }
        \end{align}
        and similarly for an $\mathcal{R}$-graded monoidal and symmetric
        monoidal comonad.
\end{definition}
\begin{definition}
        \label{defn:graded}
        An $\mathcal{R}$-\emph{graded exponential comonad} is an
        $\mathcal{R}$-graded symmetric monoidal comonad $D: (R,1,\cdot)\to
        (\SMEndoC,\Id,\circ)$ that satisfies the following additional properties:
        \begin{enumerate}
        \item  $D$ is an oplax symmetric monoidal functor $D: (R,0,+)\to
                (\SMEndoC,I,\otimes)$.  In other words,  we have monoidal
                natural transformations $e : D_0 \to I$ and $d^{m,n} : D_{m +
                n} \to D_m \otimes D_n$ making the analogues of
                \eqref{eq:gradedcomonad} for the monoidal structure
                $(\EndoC,I,\otimes)$ commute. Note that since $(R,0,+)$ is
                commutative and $D$ is symmetric the diagram
                below commutes as well.
	        \[
					\xymatrix{
						D_{m+n}\ar@{=}[r]\ar[d]_{d^{m,n}}  & D_{n+m}\ar[d]^{d^{n,m}}  \\ 
						D_m\otimes D_n\ar[r]_{\gamma} & D_n\otimes D_m
						}
	        \]
                This equips every $\catC$-object with the structure of a graded
                commutative comonoid~\cite{fujii16}.
                \item The two oplax monoidal structures of $D$ interact
                        as specified by the diagrams below (where the transformations $\phi^{D_n}$ and $\phi^{D_s}_{-,-}$ are available by virtue of the typing of $D$).
                \[
                \xymatrix@C=45pt{
                        D_{n \cdot 0} \ar[d]_{e} \ar[r]^{\delta^{n,0}} & 
                        D_n D_0 
                        \ar[d]^{D_n e}
                        \\
                        I \ar[r]_{\phi^{D_n}} & D_n I
                }
                \hspace{1.25cm}
                \xymatrix@C=45pt{
                        D_{0 \cdot s} \ar[r]^{\delta^{0,s}}
                        \ar[d]_{e} & D_0 D_s \ar[d]^{e_{D_s}} \\
                        I \ar@{=}[r] & I
                }
            	\hspace{1.25cm}
                \xymatrix@C=62pt{
                        D_{(n+m)\cdot s} \ar[r]^{\delta^{n+m,s}}
                        \ar[d]_{d^{n \cdot s, m \cdot s}} 
                        & D_{n+m} D_s \ar[d]^{d^{n,m}_{D_s}} \\
                        D_{n\cdot s} \otimes D_{m\cdot s} 
                        \ar[r]_(0.47){\delta^{n,s} \otimes \, \delta^{m,s}}
                        & D_n D_s \otimes D_m D_s
                }
        \]
        \[
                \xymatrix@C=55pt{
                        D_{s \cdot (n+m)} \ar[rr]^{\delta^{s, (n+m)}}
                        \ar[d]_{d^{(s \cdot n) + (s \cdot m)}}
                        & & D_{s} D_{n + m}
                        \ar[d]^{D_s d^{n,m}}
                        \\
                        D_{s \cdot n} \otimes D_{s \cdot m}
                        \ar[r]_(0.47){\delta^{s,n} \otimes \, \delta^{s,m}}
                        &
                        D_s D_n \otimes D_s D_m
                        \ar[r]_{\phi^{D_s}_{D_n,D_m}}
                        &
                        D_s (D_n \otimes D_m)
                }
        \]
        \end{enumerate}
\end{definition}
We now show how to interpret graded $\lambda$-calculus in an autonomous
category $\catC$ equipped with a graded exponential comonad $D$.  For every
ground type $X \in G$ we fix an interpretation $\sem{X}$ as a $\catC$-object
and interpret the type structure inductively in the usual way.  Modal types are
interpreted via the underlying graded comonad, specifically we set $\sem{\g{r}
\typeA} = D_r \sem{\typeA}$.  Given a non-empty context $\Gamma = \Gamma', x :
\typeA$, its interpretation is defined by $\sem{\Gamma', x : \typeA} =
\sem{\Gamma'} \otimes \sem{\typeA}$ if $\Gamma'$ is non-empty and
$\sem{\Gamma', x : \typeA} = \sem{\typeA}$ otherwise. The empty context is
interpreted as $\sem{-} = I$ where $I$ is the unit of $\otimes$ in $\catC$.  We
will also need some `housekeeping' morphisms to handle interactions between
context interpretation and the symmetric monoidal structure of $\catC$.  Given
contexts $\Gamma_1,\dots,\Gamma_n$ we denote by $\spl_{\Gamma_1;
\dots;\Gamma_n} : \sem{\Gamma_1, \dots, \Gamma_n} \to \sem{\Gamma_1} \otimes
\dots \otimes \sem{\Gamma_n}$ the morphism that splits $\sem{\Gamma_1, \dots,
\Gamma_n}$ into $\sem{\Gamma_1} \otimes \dots \otimes \sem{\Gamma_n}$, and by
$\join_{\Gamma_1;\dots;\Gamma_n}$ the corresponding inverse. Given a context
$\Gamma, x : \typeA, y : \typeB, \Delta$ we denote by $\exch_{\Gamma,
\underline{x : \typeA, y : \typeB}, \Delta} : \sem{\Gamma, x : \typeA, y :
\typeB, \Delta} \to \sem{\Gamma, y : \typeB, x : \typeA, \Delta}$ the morphism
corresponding to the permutation of the variable $x : \typeA$ with $y :
\typeB$.  Whenever convenient we will drop variable names in the subscripts of
$\spl$, $\join$, and $\exch$. Given a context $E \in
\Shuff(\Gamma_1;\dots;\Gamma_n)$ the morphism $\sh_E : \sem{E} \to
\sem{\Gamma_1, \dots, \Gamma_n} $ denotes the corresponding shuffling morphism.
For every sorted operation $f : \typeA_1,\dots,\typeA_n \to \typeA \in \Sigma$
we set $\sem{f} : \sem{\typeA_1} \otimes \dots \otimes \sem{\typeA_n} \to
\sem{\typeA}$ as a $\catC$-morphism. Finally we use the rules in
Figure~\ref{fig:lang_sem} to interpret judgements $\Gamma \vljud v : \typeA$ as
$\catC$-morphisms via induction over the judgement derivation system in
Figure~\ref{fig:lang}.

\begin{figure*}[h!]
        \scalebox{0.90}{
\begin{tabular}{llllr}
	& & & & \\
	\multicolumn{4}{l}
	{
	\infer[]{\sem{E \vljud f(v_1,\dots,v_n) : \typeA} = \sem{f} \comp
      (h_1 \otimes \dots \otimes h_n) \comp \spl_{\Gamma_1;\dots;\Gamma_n}
      \comp \, \sh_E}
      {\sem{\Gamma_i \vljud v_i : \typeA_i} = h_i
      \quad f : \typeA_1,\dots,\typeA_n \to \typeA \in \Sigma
      \quad E \in \Shuff(\Gamma_1; \dots ;\Gamma_n) }
	}
	&
	\infer[]{\sem{x : \typeA \vljud x : \typeA} = \id_{\sem{\typeA}}}{}
	\\
	& & & & \\
	 \infer[]{\sem{- \vljud \ast : \typeI} = \id_{\sem{\typeI}}}{}      
	 &
	 \multicolumn{4}{r}
	 {
	 \infer[]{\sem{E \vljud \prog{pm}\ v\ \prog{to}\ x \otimes y.\
          w : \typeC} =
        h \comp \join_{\Delta; \typeA; \typeB}
        \comp \, \alpha \comp \sw \comp (g \otimes \id) \comp
        \spl_{\Gamma;\Delta} \comp \sh_E}
      {\sem{\Gamma \vljud v : \typeA \otimes \typeB} = g
      \quad \sem{\Delta , x : \typeA, y : \typeB \vljud w : \typeC} = h
      \quad E \in \Shuff(\Gamma;\Delta)}
	 }
	 \\
	 & & & & \\
	 \multicolumn{2}{l}
	 {
	 \infer[]{\sem{E \vljud v \otimes w : \typeA \otimes \typeB} =
        (g \otimes h) \comp \spl_{\Gamma;\Delta}
        \comp \sh_E}{\sem{\Gamma \vljud v : \typeA} = g
      \quad \sem{\Delta \vljud w : \typeB} = h
      \quad E \in \Shuff(\Gamma;\Delta)} 
	 }
	 &
	\multicolumn{3}{r}
	{
	\infer[]{\sem{E  \vljud v \ \prog{to}\ \ast.\ w : \typeA} =
        h \comp \lambda \comp
      (g \otimes \id) \comp \spl_{\Gamma; \Delta} \comp \sh_E}
      {\sem{\Gamma \vljud v : \typeI} = g \quad \sem{\Delta 
          \vljud w : \typeA} = h \quad E \in \Shuff(\Gamma;\Delta)}
	}	 
	 \\
        & & & & \\
        \multicolumn{2}{l}
	{
        \infer[]{\sem{\Gamma \vljud \lambda x : \typeA . \, v : \typeA 
                \multimap \typeB} =
        \overline{ (h \comp \join_{\Gamma; \typeA}) }}
        {\sem{\Gamma, x : \typeA \vljud v : \typeB} = h }
        } &
	\multicolumn{3}{r}
	{
        \infer[]{\sem{E \vljud v \, w : \typeB} = \app \comp (g \otimes h)
        \comp \spl_{\Gamma;\Delta} \comp \sh_E}
        {\sem{\Gamma \vljud  v : \typeA \multimap \typeB} = g \quad
        \sem{\Delta \vljud  w : \typeA} = h \quad E \in \Shuff(\Gamma;\Delta)}     
	}
        \\ 
     \multicolumn{5}{c}{
        \dotfill 
     }
     \\
        & & & & \\
        \multicolumn{2}{l}
        {
          \infer[]{\sem{\Gamma \vljud \prog{dr}\, v : \typeA} 
          = \epsilon_{\sem{\typeA}} \comp h} 
          {\sem{\Gamma \vljud v :\ \g{1} \typeA} = h}
        }
        &
        \multicolumn{3}{r}
        {
         \infer[]{\sem{E  \vljud  \prog{ds}\, v .\ w : \typeB} =
                 h \comp \lambda \comp (e_{\sem{\typeA}} \otimes \id) \comp
        (g \otimes \id) \comp \spl_{\Gamma; \Delta} \comp \sh_E}
        {\sem{\Gamma \vljud v :\ \g{0} \typeA} = g \quad \sem{\Delta 
          \vljud w : \typeB} = h \quad E \in \Shuff(\Gamma;\Delta)}
        }
        \\
        & & & & \\
        \multicolumn{5}{c}{
                \infer[]{\sem{E \vljud \prog{cp}_{(n,m)}\
                v\ \prog{to}\ x, y.\
                u : \typeB} =
                h \comp \join_{\Delta; \typeA; \typeA}
                \comp \, \alpha \comp \sw \comp (d_{\sem{\typeA}}^{n,m} \otimes \id)
                \comp (g \otimes \id) \comp
                \spl_{\Gamma;\Delta} \comp \sh_E}
                {
                \sem{\Gamma \vljud v :\ \g{n+m}\ \typeA} = g
                \quad \sem{\Delta , x :\ \g{n} \typeA, 
                y :\ \g{m} \typeA \vljud u : \typeB} = h
                \quad E \in \Shuff(\Gamma;\Delta)}
        }
        \\
        & & & & \\
        \multicolumn{5}{c}{
                \infer[]{\sem{E \vljud \prog{pr}_{(r,\vec{s})}\ 
                                \vec{v}\ \prog{fr}\ \vec{x}.
                        \ u:\ \g{r} \typeA} 
                        = D_r h \comp D_r \join_{\typeA_1;\dots;\typeA_n}
                        \comp\, \phi_{\sem{\typeA_1},\dots,\sem{\typeA_n}}^{D_r}
                        \comp (\delta_{\sem{\typeA_1}}^{r,s_1} \otimes \dots 
                        \otimes \delta_{\sem{\typeA_n}}^{r,s_n}) \comp (g_1 \otimes 
                        \dots \otimes g_n) \comp \spl_{\Gamma_1;\dots;\Gamma_n} 
                        \comp \sh_E
                }{
                        \sem{\Gamma_i \vljud v_i :\ \g{r \cdot s_i}\typeA_i} = g_i 
                        \quad
                        \sem{x_1 :\ \g{s_1} \typeA_1,\dots, 
                        x_n :\ \g{s_n} \typeA_n \vljud u : \typeA} =h
                        \quad E \in \Shuff(\Gamma_1;\dots;\Gamma_n)  
                }
        }
\end{tabular}
}
  \caption{Judgement interpretation.}
  \label{fig:lang_sem}
\end{figure*}

The following lemma is standard and like in analogous contexts useful for
proving the soundness theorem presented below.

\begin{lemma}[Exchange and Substitution]
  \label{lem:exch_subst}
  For all judgements
  $\Gamma, x : \typeA, y : \typeB, \Delta \vljud v : \typeC$,
  $\> \Gamma, x : \typeA \vljud v : \typeB$, and
  $\Delta \vljud w : \typeA$, the following equations hold.
  \begin{align*}
    \sem{\Gamma, x : \typeA, y : \typeB, \Delta \vljud v : \typeC} & =
    \sem{\Gamma, y : \typeB, x : \typeA, \Delta \vljud v : \typeC}
    \comp \exch_{\Gamma, \underline{\typeA,\typeB}, \Delta} \\
    \sem{\Gamma, \Delta \vljud v[w / x] : \typeB} & =
    \sem{\Gamma, x : \typeA \vljud v : \typeB} \comp
    \join_{\Gamma;\typeA} \comp\, (\id \otimes \sem{\Delta \vljud w : \typeA})
    \comp \spl_{\Gamma;\Delta}
  \end{align*}
\end{lemma}

\begin{theorem}[Soundness]
  \label{theo:bsound}
  The equations presented in Figure~\ref{fig:eqs} are sound w.r.t.  judgement
  interpretation. More specifically if $\Gamma \vljud v = w : \typeA$ is one of
  the equations in Figure~\ref{fig:eqs} then $\sem{\Gamma \vljud v : \typeA} =
  \sem{\Gamma \vljud w : \typeA}$.
\end{theorem}

\section{A complete $\V$-equational system for graded $\lambda$-calculus
}\label{sec:quant_eq}

We now present a $\V$-equational system for graded $\lambda$-calculus
and prove its soundness and completeness.

\subsection{The $\V$-equational system}

We start by recalling from~\cite{dahlqvist22} the conditions imposed on $\V$ to
obtain a well-behaved framework of $\V$-equations. We will then extend this
framework to the graded seting.  Let $\mathcal{V}$ denote a commutative and
unital quantale, $\otimes : \mathcal{V} \times \mathcal{V} \to \mathcal{V}$ the
corresponding binary operation, and $k$ the unit~\cite{paseka00}.  Consider  now
the two following definitions concerning ordered
structures~\cite{GHK+03,JGL-topology} (they will allow us to work with 
specified subsets of $\mathcal{V}$-equations chosen \eg for computational
reasons~\cite{dahlqvist22}).

\begin{definition}
	Take a complete lattice $L$.  For every $x, y \in L$ we say that
	$y$ is \emph{way-below} $x$ (in symbols, $y \ll x$) if for every
	subset $X \subseteq L$ whenever $x \leq \bigvee X$ there exists a
	\emph{finite} subset $A \subseteq X$ such that $y \leq \bigvee A$.
	The lattice $L$ is called \emph{continuous} iff for every $x \in L$,
	\begin{flalign*}
		x = \bigvee \{ y  \mid y \in L\ \text{and} \ y \ll x \}
	\end{flalign*}
\end{definition}

\begin{definition}
	Let $L$ be a complete lattice. A \emph{basis} $B$ of $L$ is a subset
	$B \subseteq L$ such that for every $x \in L$ the set
	$B \cap \{ y \mid y \in L\ \text{and} \ y \ll x \}$ is directed and
	has $x$ as the least upper bound.
\end{definition}
We assume that the underlying lattice of $\mathcal{V}$ is continuous and has a
basis $B \ni k$ closed under finite joins and multiplication. As alluded above,
the continuity condition will allow us to work only with $\V$-equations whose
label is in $B$.  We also assume that $\mathcal{V}$ is \emph{integral}, \ie
that the unit $k$ is the top element of $\mathcal{V}$, a common assumption in
quantale theory~\cite{velebil19} that facilitates some of our results.

\begin{example}
        The Boolean quantale $((\{0 \leq 1\}, \vee), \otimes := \wedge)$ is
        \emph{finite} and thus continuous~\cite{GHK+03}. Since it is
        continuous, $\{0,1\}$ itself is a basis for the quantale that satisfies
        the conditions above. For the metric quantale $(([0,\infty], \wedge),
        \otimes := +)$ (note that the order on this quantale is the opposite of the usual order on $[0,\infty]$), the way-below relation corresponds to the
        \emph{strictly greater} relation with $\infty > \infty$, and a basis
        for the underlying lattice that satisfies the conditions above is the
        set of extended non-negative rational numbers. Other examples of
        quantales that satisfy the conditions above can be found
        in~\cite{dahlqvist22}.
\end{example}

A $\mathcal{V}$-equation-in-context is an expression $\Gamma \vljud v =_q w :
\typeA$ where $q \in B$ (the basis of $\mathcal{V}$), and $\Gamma \vljud v :
\typeA$, $\Gamma \vljud w : \typeA$ are graded $\lambda$-terms. If $\V$ is the
metric quantale we obtain metric equations-in-context and if $\V$ is the
Boolean quantale we obtain inequations-in-context (where $v =_1 w$ corresponds
to $v \leq w$).  In this $\V$-equational setting a classical
equation-in-context $v = w$ translates to $v =_k w \wedge w =_k v$. For example
in the metric case $v = w \equiv v =_0 w \wedge w =_0 v$ and in the Boolean
case $v = w \equiv v \leq w \wedge w \leq v$.

We can now move to the graded setting.
\newcommand{\smu}{\mathbin{\bullet}}

\begin{definition}\label{defn:scalar}
        A \emph{scalar multiplication} of a semiring $\mathcal{R}$ on a
        quantale $\V$ is a function $\smu : R \times \V \to \V$ such that for
        each $k\in R$, the map $k\smu -:\V\to\V$ preserves joins in $\V$.
\end{definition}
The definition entails in particular that for all $v,v' \in \V$ if $v \geq v'$
then $k \smu v \geq k \smu v'$. 
\begin{definition}[Graded $\V \lambda$-theories]\label{defn:theory}
  Consider a tuple $(G,\Sigma)$ consisting of a set $G$ of ground types and a
  set $\Sigma$ of sorted operation symbols.  A \emph{graded $\V
  \lambda$-theory} $((G,\Sigma),Ax)$ is a triple such that $Ax$ is a set of
  $\V$-equations-in-context between $\lambda$-terms built from $(G,\Sigma)$.
\end{definition}
The elements of $Ax$ are called the \emph{axioms} of the theory. Let $Th(Ax)$
be the smallest $\V$-indexed binary relation (the $\V$-equations) that contains
$Ax$, the equational schema presented in Figure~\ref{fig:eqs},  and that is
closed under the rules listed in Figure~\ref{fig:theo_rules}. We call the
elements of $Th(Ax)$ the \emph{theorems} of the theory.  Intuitively the rules
in Figure~\ref{fig:theo_rules} above the first dotted line can be seen as a
$\V$-generalisation of an equivalence relation (see~\cite{dahlqvist22} for a
more detailed explanation). The other rules correspond to a $\V$-generalisation
of compatibility. The rule concerning promotion is slightly different from the
others in that it involves a $k$-factor ($k \smu -$) to reflect the fact that
$u$ (resp. $u'$) becomes usable $k$-times.  Finally, note that we can
consider \emph{symmetric} graded $\V \lambda$-theories by adding to the mix the
rule,
\[
        \infer{w =_q v}{v =_q w}
\]
This is desirable for example in the (ultra-)metric case but makes no sense if
one wishes to work with inequations (graded inequational
$\lambda$-theories collapse to graded equational ones under this rule).
\begin{figure}[h!]
        \centering
        \scalebox{0.905}{
	\begin{tabular}{l l l r}
		\infer[\textbf{(refl)}]{v =_\top v}{}
		&
                \multicolumn{2}{c}{
		\infer[\textbf{(trans)}]{v =_{q_1 \otimes q_2} u}{
			v =_{q_1} w  \qquad
			w =_{q_2} u}
                }
		&
                \infer[\textbf{(weak)}]{v =_{q_2} w}{v =_{q_1} w \qquad {q_2} \leq {q_1} }
		\\
                & & &  \\
                \multicolumn{2}{r}{
		\infer[\textbf{(arch)}]{v =_{q_1} w}{
			\forall {q_2} \ll q_1 .\ v =_{q_2} w}
                }
                &
                \multicolumn{2}{c}{
		        \infer[\textbf{(join)}]{v =_{\vee q_i} w}{
                        \forall i \leq n.\ v =_{q_i} w}
                }
                \\
                \multicolumn{4}{c}{
                        \dotfill 
                }
                \\
                & & & \\
		\infer[]{f(v_1,\dots,v_n) =_{\otimes q_i} f(w_1,\dots,w_n)}
		{\forall i \leq n.\ v_i =_{q_i} w_i}
		&
                \multicolumn{2}{r}{
		\infer[]{
			\prog{pm}\ v\ \prog{to}\ x \otimes y.\ v' =_{q_1 \otimes q_2}
			\prog{pm}\ w\ \prog{to}\ x \otimes y.\ w'}
			{v =_{q_1} w \qquad v' =_{q_2} w' }
		}
                &
                \infer[]{
			v\ \prog{to}\ \ast .\  v'=_{q_1 \otimes q_2}
			w\ \prog{to}\ \ast .\ w'}
		{v =_{q_1} w \qquad v' =_{q_2} w'}
		\\
                &&&
                \\
                \infer[]{v \otimes v' =_{q_1 \otimes q_2} w \otimes w'}{
			v =_{q_1} w \quad v' =_{q_2} w'}
               	&
                \multicolumn{2}{c}{
                \infer[]{
                        \lambda x : \typeA .\ v =_q  \lambda x :
                              \typeA .\ w}{v =_q w}
                }
                &
                \infer[]{v \, v' =_{q_1 \otimes q_2} w \, w'}
		        {v =_{q_1} w \quad v' =_{q_2} w'}
                \\
                \multicolumn{4}{c}{
                        \dotfill 
                }
                \\
                &&&
                \\
                \infer[]{\prog{dr} \,v =_q \prog{dr} \,v'}
                {v =_q v'}
                &
                \multicolumn{2}{c}{
                  \infer{\prog{cp}_{(n,m)} v\ \prog{to}\ x,y.\ w =_{q_1 \otimes q_2}
                  \prog{cp}_{(n,m)} v'\ \prog{to}\ x,y.\ w'}{
                       v =_{q_1} v' \quad w =_{q_2} w' 
                }}
                &
                \infer[]{\prog{ds} \,v.\ w =_{q_1 \otimes q_2} \prog{ds} \,v'.\ w'}
                { v =_{q_1} v' \quad w =_{q_2} w' }
                \\
                &&&
                \\
                \infer[]{\Delta \vljud v =_q w : \typeA}{
	 	 		\Gamma \vljud v =_q w : \typeA \qquad \Delta \in
                 \mathrm{perm}(\Gamma)}
                 &
                \multicolumn{2}{c}{
                    \infer[]{\prog{pr}_{(r,\vec{s})}\ \vec{v}\ \prog{fr}\
                         \vec{x} .\ u =_{\otimes q_i \otimes (r \smu q')}
                         \prog{pr}_{(r,\vec{s})}\ \vec{v'}\ \prog{fr}\
                                \vec{x} .\ u'
                     }{\forall i \leq n.\ v_i =_{q_i} v'_i \qquad u =_{q'} u'}}
                 &
                 \infer[]{v[v'/x] =_{q_1 \otimes q_2}w[w'/x]}
	 	  {v =_{q_1} w \qquad v' =_{q_2} w'}		
                 \end{tabular}
}
\caption{$\mathcal{V}$-congruence rules.}
\label{fig:theo_rules}
\end{figure}

\subsection{Interpretation of $\V$-equations, soundness, and completeness}
\label{sec:veq}
In this subsection we recall the interpretation of $\V$-equations in the
setting of linear $\lambda$-calculus~\cite{dahlqvist22} and extend it to the
graded case. The main idea is that we suitably enrich the interpretation
structure in Definition~\ref{defn:graded} (an autonomous category equipped with
a graded exponential comonad) so that the corresponding hom-sets become
equipped with a `generalised metric structure'.  More technically the basis of
enrichment is that of
$\V$-categories~\cite{lawvere73,stubbe14,hofmann20,velebil19}, a concept which
we recall below.  We prove soundness and completeness of the previous
$\V$-equational system \wrt this interpretation.

\begin{definition}
  \label{defn:vcat}
  A $\mathcal{V}$-category is a pair $(X,a)$ where $X$ is a set and $a : X
  \times X \to \mathcal{V}$ is a function that satisfies $k \leq a(x,x)$ and
  $a(x,y) \otimes a(y,z) \leq a(x,z) $ for all $x,y,z \in X$.  For two
  $\mathcal{V}$-categories $(X,a)$ and $(Y,b)$, a $\mathcal{V}$-functor $f :
  (X,a) \to (Y,b)$ is a function $f : X \to Y$ that satisfies the inequality
  $a(x,y) \leq b(f(x),f(y))$ for all $x,y \in X$.
\end{definition}
$\mathcal{V}$-categories and $\mathcal{V}$-functors form a category which we
denote by $\VCat$.  A $\mathcal{V}$-category $(X,a)$ is called \emph{symmetric}
if $a(x,y) = a(y,x)$ for all $x,y \in X$. We denote by $\VCatSy$ the full
subcategory of $\VCat$ whose objects are symmetric. Every
$\mathcal{V}$-category carries a natural order defined by $x \leq y$ whenever
$k \leq a(x,y)$. A $\mathcal{V}$-category is called \emph{separated} if its
natural order is anti-symmetric. We denote by $\VCatSe$ the full subcategory of
$\VCat$ whose objects are separated.
\begin{example}
  For $\mathcal{V}$ the Boolean quantale, $\VCatSe$ is the category $\Pos$ of
  partially ordered sets and monotone maps, and $\VCatSS$ is the category
  $\Set$ of sets and functions.  For $\mathcal{V}$ the metric quantale,
  $\VCatSS$ is the category $\Met$ of metric spaces and non-expansive maps.
  For more examples see~\cite{dahlqvist22}.
\end{example}
We will take advantage of the following useful facts about $\V$-categories.
The inclusion functor $\VCatSe \hookrightarrow \VCat$ has a left adjoint
\cite{hofmann20}. It is constructed first by defining the equivalence relation
$x \sim y$ whenever $x \leq y$ and $y \leq x$ (where $\leq$ is the natural
order introduced earlier). Then this relation induces the separated
$\mathcal{V}$-category $(X/_\sim, \tilde a)$ where $\tilde a$ is defined as
$\tilde a([x],[y]) = a(x,y)$ for every $[x],[y] \in X/_\sim$. Finally the left
adjoint of the inclusion functor $\VCatSe \hookrightarrow \VCat$ sends every
$\mathcal{V}$-category $(X,a)$ to $(X/_\sim, \tilde a)$.  The category
$\VCat$ is autonomous with the tensor $(X,a) \otimes (Y,b) := (X \times Y, a
\otimes b)$ where $a\otimes b$ is defined as $(a \otimes b)((x,y), (x',y')) =
a(x,x') \otimes b(y,y')$ and the set of $\mathcal{V}$-functors
$\VCat((X,a),(Y,b))$  equipped with the map,
\[
(f,g) \mapsto \bigwedge_{x \in X} b(f(x),g(x))
\]
$\VCatSy$, $\VCatSe$, and $\VCatSS$
inherit the autonomous structure of $\VCat$ whenever $\mathcal{V}$ is
integral~\cite{dahlqvist22}.

\begin{definition}\label{defn:enr_aut}
  A $\VCat$-enriched autonomous category $\catC$ is an autonomous and
  $\VCat$-enriched category $\catC$ such that the bifunctor $\otimes : \catC
  \times \catC \to \catC$ is a $\VCat$-functor and the adjunction $(- \otimes
  X) \dashv (X \multimap -)$ is a $\VCat$-adjunction.  We obtain analogous
  notions of enriched autonomous category by replacing $\VCat$ (as basis of
  enrichment) with $\VCatSe$, $\VCatSy$, or $\VCatSS$.
\end{definition}

\begin{example}\label{ex:pos_met}
  The categories $\Pos$, $\Met$, and $\Set$ are instances of
  Definition~\ref{defn:enr_aut}. 
\end{example}
We now turn our attention to the graded case, more specifically on how to
suitably enrich the underlying graded exponential comonad. An obvious way of
doing so would be to state that for every $r \in R$ the functor $D_r : \catC
\to \catC$ is $\VCat$-enriched.  This however turns out to be too strict to
soundly interpret the $\V$-compatibility rule concerning promotion
(Figure~\ref{fig:theo_rules}).  Instead we adopt a more relaxed variant which
formally resembles the well-known notion of Lipschitz-continuity from calculus.

\begin{definition}\label{def:Lipschitz}
        An $\mathcal{R}$-Lipschitz exponential comonad (for a scalar multiplication $\smu: R\times \V\to\V$) is 
        an $\mathcal{R}$-graded exponential comonad
        such that 
        the inequality, \[ r\smu a(f,g)
        \leq a(D_r f, D_r g) \] holds for all $\catC$-morphisms $f,g : X \to Y$
        and $r \in R$.
\end{definition}

\begin{definition}[Models of graded $\V\lambda$-theories]\label{defn:model}
        Consider a  graded $\V\lambda$-theory $((G,\Sigma),Ax)$ and a
        $\VCatSe$-autonomous category $\catC$ equipped with an $\mathcal{R}$-Lipschitz
        exponential comonad. Suppose that for each $X \in G$ we have an
        interpretation $\sem{X}$ as a $\catC$-object and analogously for the
        operation symbols.  This interpretation structure is a \emph{model} of
        the theory if all axioms are satisfied by the interpretation, \ie if
        $v =_q w$ is an axiom of the theory then $a(\sem{v},\sem{w}) \geq q$.

        In the case of \emph{symmetric} graded $\V \lambda$-theories the
        corresponding notion of a model is obtained by replacing the basis of
        enrichment (\ie $\VCatSe$) by $\VCatSS$.
\end{definition}

We can now prove that the $\V$-equational system of graded $\lambda$-calculus
is sound and complete \wrt Definition~\ref{defn:model}.

\begin{theorem}[Soundness]
        \label{theo:sound}
        Consider a (symmetric) $\V \lambda$-theory $\mathscr{T}$ and a model
        $M$ of $\mathscr{T}$ over $\catC$. If $v =_q w$ is a theorem of $\mathscr{T}$
        then $a(\sem{v},\sem{w}) \geq q$.
\end{theorem}
\begin{proof}
        The fact that the equational schema listed in Figure~\ref{fig:eqs} is
        sound follows from Theorem~\ref{theo:bsound} and the definition of a
        $\V$-category (Definition~\ref{defn:vcat}). The proof then follows by
        induction over the rules listed in Figure~\ref{fig:theo_rules}.  We
        only focus on those rules that concern graded modal types (the other
        ones were already proved in~\cite{dahlqvist22}). The case of
        deriliction follows directly from the fact that for all $X \in |\catC|$
        the morphism $\epsilon_X : D_1 X \to X$ lives in $\catC$ and $\catC$ is
        $\VCat$-enriched. The rules that concern copying and discarding follow
        from an analogous reasoning.  The rule that concerns promotion also
        follows similarly to the above except that we use the two following
        properties: first, for all $q,q'\in \V$ and $r \in R$ if $q \geq q'$
        then $r \smu q \geq r \smu q'$;  second, the fact that the graded
        comonad is Lipschitz. In conjunction both properties entail the
        implication $a(\sem{u},\sem{u'}) \geq q' \Rightarrow a(D_r \sem{u},D_r
        \sem{u'}) \geq r \smu q'$.
\end{proof}

The completeness result is based on the idea of a \emph{Lindenbaum-Tarski}
algebra: it follows from building the syntactic category $\Syn(\mathscr{T})$ of
$\mathscr{T}$, showing that it is a model of $\mathscr{T}$, and then showing
that if $a(\sem{v},\sem{w}) \geq q$ in $\Syn(\mathscr{T})$ the $\V$-equation $v
=_q w$ is a theorem of $\mathscr{T}$.  In order to build $\Syn(\mathscr{T})$
and to show that it is indeed a model of $\mathscr{T}$, we resort to the notion
of a \emph{multicategory} and associated
constructions~\cite{lambek69,leinster04,hermida00,lobbia21}. More specifically,
we will first generate a syntactic multicategory $\Syn_M(\mathscr{T})$ from
$\mathscr{T}$ and then show that the former induces an autonomous
$\Syn(\mathscr{T})$ with the necessary requisites to be a model of
$\mathscr{T}$. The reason we involve multicategories is that some equations we
need to face are much more easily proved in this framework, an observation
already made in analogous contexts~\cite{lambek69,benton92}.  For the same
purpose, we also use a bijective correspondence between graded
comonads and graded co-Kleisli triples on a multicategory.

\begin{theorem}[Soundness \&  Completeness]
  \label{theo:compl} For a (symmetric) graded $\V \lambda$-theory
  $\mathscr{T}$,  a $\V$-equation $\Gamma \vljud v =_q w : \typeA$ is a theorem
  of $\mathscr{T}$ iff it is satisfied by all models of the theory.
\end{theorem}
\section{A canonical construction of Lipschitz exponential
comonads}\label{sec:canonical}

This section presents a canonical construction of Lipschitz exponential
comonads on $\VCat$-autonomous categories that satisfy certain conditions.  The
construction is inspired by~\cite{mellies09}, which shows how to build
(non-graded) exponential comonads via the notion of a (co)free commutative
(co)monoid. In order to describe the connection to \emph{op.\ cit.}\ at a
suitable level of abstraction, we start with a brief overview of this
construction in the form of abstract categorical results. We will then provide a more direct construction.

Let $\catfont{Mon}_\pi(\catC)$ be the category of commutative monoids in a
symmetric monoidal category $\catC$.  A crucial observation is that a comonoid
in $\catC$ is the same thing as a monoid in $\catC^\cop$~\cite{mellies09} --
thus the category of commutative comonoids can be seen as
$\catfont{Mon}_\pi(\catC^\cop)^\cop$.  The other relevant key observation is
that the forgetful functor $\catfont{Mon}_\pi(\catC) \to \catC$ is right
adjoint if conditions concerning the existence and preservation of a certain
limit are met (\cf \cite{mellies09}). By duality this induces a forgetful
functor $\catfont{Mon}_\pi(\catC^\cop)^\cop \to \catC$ which is furthermore
\emph{left} adjoint. Such an adjoint situation induces a comonad on $\catC$
which can be shown to be exponential (see~\cite{mellies09}).
Now, we are interested in extending these ideas to the graded setting with
$\mathcal{R}$ as the semiring of natural numbers. To this effect we recall next
the notion of a strict action.

\begin{definition}
        Let $\catM$ be a monoidal category and $\catC$ be an arbitrary
        category. A strict action is a functor $\oast : \catM \times \catC \to
        \catC$ that satisfies the following equations for all $\catM$-objects
        $m,n$ and $\catC$-objects $X$:
        \[
                X = I \oast X 
                \hspace{3cm}
                m \oast (n \oast X) = (m \otimes n) \oast X
        \]
\end{definition}
Consider then both a strict action $\oast : \catM \times \catA \to \catA$,
where $\catM$ is a discrete category, and an adjoint situation $L \dashv R :
\catA \to \catB$.  It is well-known that both constructions yield an
$\catM$-graded monad on $\catB$ with $T_n (X) = R (n \oast L X)$ (see details
in~\cite{fujii16}). This is the basis to extend~\cite{mellies09} to a graded
setting.

Specifically let $\Nats$-$\catfont{Mon}_\pi(\catC)$ be the category of
$(\Nats,+,0)$-graded commutative monoids in $\catC$.  Following an analogous
reasoning to the previous paragraphs, one may regard
($\Nats$-$\catfont{Mon}_\pi(\catC^\cop))^\cop$ as the category of
$(\Nats,+,0)$-graded commutative comonoids in $\catC$. There is also a
forgetful functor $(-)_1 : \Nats$-$\catfont{Mon}_\pi(\catC) \to \catC$ which
given a graded monoid only keeps the $1$-component of the underlying carrier.
Then under mild conditions, also pertaining to the existence and preservation
of a certain limit $\lim \Diag$ (details below), this functor is right adjoint. And
thus in particular
$(-)_1$ : ($\Nats$-$\catfont{Mon}_\pi(\catC^\cop))^\cop \to \catC$ is left
adjoint. Finally via a few routine calculations one can show the existence of a
strict action $\oast : (\Nats, \, \cdot \,,1) \times
\Nats$-$\catfont{Mon}_\pi(\catC) \to$ $\Nats$-$\catfont{Mon}_\pi(\catC)$
defined by, 
\[
        (k,((X_n)_{n \in \Nats}, e, f_{m,n} : X_m \otimes X_n \to X_{m+n}))
        \mapsto
        ((X_{n \cdot k})_{n \in \Nats}, e, 
        f_{m \cdot k, n \cdot k} :
        X_{m \cdot k} \otimes X_{n \cdot k} \to X_{(m+n) \cdot k})
\]
Together with the previous adjoint situation this yields  an $(\Nats, \, \cdot
\,, 1)$-graded comonad on $\catC$. By unfolding the respective definitions one
can show that this comonad is that of \emph{symmetric powers} described in a
very recent publication~\cite{lemay23} and stated to be exponential. Due to space constraints
we describe only the functorial component.  Subsequently we will show that this comonad is
Lipschitz under the condition that the aforementioned limit of $\Diag$ is
$\VCat$-enriched.

As an instructive first approximation of the $\N$-Lipschitz exponential comonad we intend to describe, consider the map,
\begin{align}\label{eq:defD}
        D: \Nats\to [\catC,\catC], n\mapsto \underbrace{\Id\otimes
        \ldots\otimes\Id}_{n\text{times}}.
\end{align}
The assignment $D$ \emph{almost} defines a canonical $(\Nats, \cdot, 1)$-graded
exponential comonad. 
\begin{theorem}\label{theo:D}
        The assignment $D$ of \eqref{eq:defD} satisfies all the conditions of
        Definition~\ref{defn:graded} except for symmetry in condition (i)
        and the last diagram of condition (ii).
\end{theorem}

In order to construct an exponential comonad on $\catC$ one needs to remedy the
lack of symmetry of $D$. To do this, one can consider the sub-$\N$-graded
comonad of $D$ which only keeps the \emph{symmetric elements} in the tensor
products $D_n X=X^{\otimes n}$. For this we follow the second step of the
construction in~\cite{mellies09}. Every element $\sigma$ in the permutation
group $\Perm$ on $n$ elements defines a natural transformation $D_n\to D_n$
which we also denote by $\sigma$. We now define $E: \N\to [\catC,\catC]$ by
mapping $n\in\N$ to the limit $E_n$ of the diagram~\eqref{eq:Eobj} defined by all these natural transformations. 
Each $E_n$ is defined on morphisms
in the obvious way: if $f:X\to Y $ is a $\catC$-morphism then since
$D_nf\comp\sigma=\sigma\comp D_n f$, the universal property of $E_nY$
guarantees the existence of a unique $\catC$-morphism $E_n f$ that makes Diagram
\eqref{eq:Emor} commute.

\noindent\begin{minipage}{.4\linewidth}
	\begin{equation}\label{eq:Eobj}
		\xymatrix@C=50pt{
			E_n\ar[r]^{\eps} & D_n \ar@<5pt>[r]^(0.49){\sigma\in \Perm}_(0.49){}
			\ar@<-5pt>[r]_(0.48){\tau\in\Perm}
			\ar@{{}{ }{}}[r]|-(0.48){\dots}
			& D_n
		}
	\end{equation}
\end{minipage}%
\begin{minipage}{.6\linewidth}
	\begin{equation}\label{eq:Emor}
		\xymatrix@C=50pt{
			E_n X\ar[r]^{\eps_X} \ar@{..>}[d]_{E_nf}
			& D_n X\ar@<5pt>[r]^(0.49){\sigma\in\Perm}_(0.49){}
			\ar@<-5pt>[r]_(0.48){\tau\in\Perm}
			\ar@{{}{ }{}}[r]|-(0.48){\dots} \ar[d]_{D_n f}
			& D_n X\ar[d]^{D_n f} 
			\\
			E_n Y\ar[r]_{\eps_Y} 
			& D_n Y\ar@<5pt>[r]^(0.49){\sigma\in\Perm}_(0.49){}
			\ar@<-5pt>[r]_(0.48){\tau\in\Perm}
			\ar@{{}{ }{}}[r]|-(0.48){\dots} & D_n Y 
		}
	\end{equation}
\end{minipage}

\begin{theorem}\label{theo:E}
        Suppose that for every $\catC$-object $X$, $(X\otimes -)$ preserves the
        limits~\eqref{eq:Eobj}. Then the assignment $E$ defined by the
        limits~\eqref{eq:Eobj} induces a sub-$\N$-graded comonad of $D$ which is
        furthermore an $\N$-graded exponential comonad.
\end{theorem}

We will now show that the graded comonad $E$ is additionally Lipschitz. First
we define the following scalar multiplication.
\begin{proposition}\label{prop:cansmu}
        For any commutative quantale $\V$, the map $\smu: \N\times\V\to\V$
        defined by, \[
	n\smu q = \underbrace{q\otimes \ldots\otimes q}_{n\text{ times}}\text{ if }n\neq 0\qquad 0\smu q= k
	\]
	is a scalar multiplication in the sense of Definition~\ref{defn:scalar}.
\end{proposition}

\begin{proof}
        To see that $n\smu - $ preserves arbitrary joins we compute,
	\begin{flalign*}
                & \, n \smu \left (\bigvee X \right ) & \\
		& \defeq \left (\bigvee X \right ) \otimes \dots \otimes \left (\bigvee X \right ) & \\
		& = \bigvee \left (X \otimes \dots \otimes X \right ) &
		\text{\{$\otimes$ preserves joins\}} \\
		& = \bigvee \{ x_1 \otimes \dots \otimes x_n \mid x_1,\dots,x_n \in X \} & \\
		& = \bigvee \{ x \otimes \dots \otimes x \mid x \in X \} & 
		\text{$\{\star\}$} \\
		& \defeq \bigvee n \smu X &
	\end{flalign*}
        where the step marked with $(\star)$ follows from the fact that the
        inequation below holds.
	\[
	x_1 \otimes \dots \otimes x_n \leq \left
	( \bigvee \{x_1,\dots,x_n\} \right ) \otimes \dots \otimes \left ( \bigvee
	\{x_1,\dots,x_n\} \right )
        \hspace{1.5cm}
        (x_1,\dots,x_n \in X)
	\]
\end{proof}

Next, let $\catC$ be a $\VCat$-autonomous category and the underlying diagram
of~\eqref{eq:Eobj} for a $\catC$-object $X$ be denoted by $\Diag$. Also assume that for every two cones
$f,g : A \to X^{\otimes n}$ for $\Diag$ the equation $a(f,g) = a(f', g')$
holds where $f', g' : A \to E_n(X)$ are the corresponding mediating morphisms.
More compactly this amounts to the statement that $\catC$ has the $\VCat$-limit of
$\Diag$ weighted by the functor $!$ (constant on the $\VCat$-object $1$).  
This condition guarantees that $E$ is $\N$-Lipschitz.

\begin{theorem}
        \label{theo:Lip}
        Consider a $\VCat$-autonomous category $\catC$ such that it has the
        $\VCat$-limit of $\Diag$ weighted by $!$ and additionally assume that for every
        $\catC$-object $X$ the functor $(X \otimes -)$ preserves this limit, then $E$ is an 
        $\N$-Lipschitz exponential comonad.
\end{theorem}
\begin{proof}
        Consider two $\catC$-morphisms $f,g : X \to Y$.  We reason,
	\begin{flalign*}
                & \, n\smu a(f,g) &  \\
		& \defeq \underbrace{a(f,g) \otimes \ldots\otimes 
                a(f,g)}_{n\text{ times}} & \\
		&\leq a(f^{\otimes n}, g^{\otimes n}) &
                \{\text{$\otimes$ in $\catC$ is $\VCat$-enriched}
                        \} \\
		&\defeq a(D_n f, D_n g)&\\
		&\leq a(D_nf\comp \eps_X, D_n g\comp \eps_X)&
                \{\text{$\catC$ is $\VCat$-enriched}\} \\
                &= a(E_n f, E_n g)&\text{\{limit of $\Diag$ is $\VCat$-enriched\}}
	\end{flalign*}
\end{proof}

\section{Applications to timed and probabilistic computation}\label{sec:examples}

\newcommand{\Dil}{\mathrm{Dil}}

\subsection{Timed computation and dilations}
We now revisit the example of wait calls from \S\ref{sec:intro}
and equip it with a concrete model by applying the canonical
construction of $\N$-Lipschitz exponential comonads detailed in
\S\ref{sec:canonical}.  Recall that the example is based on a ground type $X$ and a
signature $\{ \prog{wait_n} : X \to X \mid n \in \Nats \}$ of wait calls.
Consider then the following metric axioms proposed in~\cite{dahlqvist22}:
\begin{flalign}\label{ax} 
        \prog{wait_0}(x) =_0
        x \hspace{1cm} \prog{wait_n}(\prog{wait_m}(x)) =_0 \prog{wait_{n + m
                }}(x) \hspace{1cm} \infer{\prog{wait_n}(x) =_\epsilon
        \prog{wait_m}(x)}{\epsilon = |m - n|} 
\end{flalign} 
In order to apply the construction in \S\ref{sec:canonical},  we need first of
all a $\Met$-enriched autonomous category. For this case we choose
$\Met$ itself (\emph{cf.}\ Example~\ref{ex:pos_met}).  Next we show that the
tensor $\otimes$ in $\Met$ preserves all limits; actually we prove the
following more general claim.
\begin{proposition}\label{prop:tensor_preservation}
       Let $\V$ be a quantale whose operation $\otimes$ preserves arbitrary
       meets and let us consider the respective category $\VCat$. For every
       $\V$-category $X$ the functor $(- \otimes X) : \VCat \to \VCat$
       preserves all limits.  The same property holds for the cases $\VCatSe$,
       $\VCatSy$, and $\VCatSS$.
\end{proposition}
\begin{corollary}
       For all categories $\catC$ mentioned in Example~\ref{ex:pos_met} (which
       includes $\Met$) and $\catC$-objects $X$ the functor $(- \otimes X) :
       \catC \to \catC$ preserves all limits.
\end{corollary}
Finally, it is straightforward to prove that $\Met$ has the $\VCat$-limit of
$\Diag$ weighted by $!$ and therefore all pre-requisites of the construction
are satisfied. By unfolding the respective definitions we deduce that $E_n(X)$
is the metric space whose elements are $n$-copies $(x,\dots,x)$ of an element
$x \in X$ and whose metric is the restriction of the metric in $X^{\otimes n}$.
The counit is the identity and comultiplication amounts to rebracketing. The
operation $d^{m,n}$ amounts to rebracketing as well.  It is then easy to build
a model for the metric theory of wait calls that was previously presented: fix
a metric space $A$, interpret the ground type $X$ as $\Nats \otimes A$ and the
operation symbol $\prog{wait_n} : X \to X$ as the non-expansive map $
\sem{\prog{wait_n}} : \Nats \otimes A \to \Nats \otimes A, (i,a) \mapsto (i
+ n, a)$.  It only remains to prove that the axioms in~\eqref{ax} are satisfied
by the proposed interpretation, but this can be shown via a few routine
calculations.

We end this subsection by relating the comonad that we canonically obtained to
the comonad of dilations presented in~\cite{katsumata18}. The latter's main
idea is that of distance dilation: given a metric space $(X,d)$ we obtain a new
one $\mathrm{Dil}_n(X,d) := (X, n \smu d)$ by scaling up distances via multiplication, more
concretely $(n \smu d)(x,y) = n \smu d(x,y)$ for $n \in \Nats$
and $x,y \in X$.  It is easy to see that $E_n \cong \mathrm{Dil}_n$ and
moreover that the underlying comonadic operations agree.  It is also  easy to
see that the copy, discard, and monoidal operations agree as well.  This yields
the following result.

\begin{corollary}
        \label{cor:met}
        The $\Met$-autonomous category $\Met$ of metric spaces and
        non-expansive maps equipped with the comonad of dilations
        yields a model of the metric theory of wait calls~\eqref{ax}.
\end{corollary}

\subsection{Probabilistic computation}
\newcommand{\nat}{\prog{nat}}
\newcommand{\real}{\prog{real}}
\newcommand{\normal}{\prog{normal}}
\newcommand{\meas}{\mathcal{M}}
\newcommand{\bern}{\mathtt{bernoulli}}
\newcommand{\ptp}{\widehat{\otimes}_{\pi}}

\cite[Example 28]{dahlqvist22} presents a metric equational system to reason
about the total variation distance between  distributions constructed as
probabilistic programs, specifically individual steps in non-standard random
walks. It is however cumbersome to reason about distances between random walks
consisting of $n$ steps when they are expressed in a purely linear language.
This is because a probabilistic term like $\normal(0,1)$ operationally
corresponds to a \emph{single} sample which cannot be copied. Thus, to write a
program using $n$ normal deviates we need to call $n$ \iid samples from
$\underbrace{\normal(0,1)\otimes \ldots\otimes \normal(0,1)}_{n\text{ times}}$
which is inconvenient and unclear (especially for large values of $n$), but also
difficult to maintain and generalise. Using a graded system, we can not only
assume a clean and parametric access to such \iid samples but also to more
complex sampling schemes (details below). Furthermore, we have a convenient way
of manipulating such sequences of samples via the promotion rule, and to feed
them into $n$-ary functions through the copy (\ie contraction) rule. All of
this whilst maintaining the ability to reason about distances between programs.

Let us illustrate our previous remarks with some simple examples.  We start by
briefly presenting a toy probabilistic language (more details can be found
in~\cite{dahlqvist22}).  We consider only two ground
types $\real$ and $\real^+$ (in particular, we will view the integers 0 and 1 
as reals). The graded modal type $\g{n} \real$ can then be thought of as the
type of $n$ real samples. We also consider a signature of operations consisting
of the real numbers $\{ r : \typeI \to \real \mid r \in \mathbb{Q} \}$, the
addition and multiplication operations $+,\ast: \real,\real\to\real$, and
finally three collections of built-in samplers which we detail next.  The first
collection consists of samplers returning $k$ samples from an urn containing
$m$ balls labelled 0 and $n$ balls labelled 1 \emph{with replacement} (\ie we
return the ball to the urn after reading its value).  We denote the samplers of
this class $\prog{replace}(k,m,n): \, \g{k} \real$. The second collection
samples from the same urn model but \emph{without replacement}. We denote these
samplers $\prog{no\_replace}(k,m,n): \, \g{k} \real$ (and of course require
that $k\leq m+n$). The third class $\prog{iid\_normal}(k;\mu,\sigma):
\real,\real^+ \to \, \g{k} \real$ will simply sample $k$ \iid normal deviates. 

We proceed by providing a concrete graded $\lambda$-model for the language.
First we fix the category $\Ban$ of Banach spaces and linear contractions as
our $\Met$-enriched autonomous category (see~\cite{K79c,DK20a,DSK20a,dahlqvist22} for
more details about this style of semantics). Specifically $\Ban$ is autonomous
when equipped with the projective tensor product $\hat \otimes_\pi$ and the
internal hom $\multimap$ defined as the space of \emph{bounded} linear maps
equipped with the sup-norm~\cite{ryan13}. It is also straightforward to prove
that $\Ban$ has the $\Met$-limit of $\Diag$ weighted by $!$. Then in order to
apply the construction in $\S\ref{sec:canonical}$ we use the following result.
\begin{proposition}\label{prop:BanTensor}
        For every Banach space $W$ and every $n\in\N$, the functor $(- \otimes
        W) : \Ban \to \Ban$ preserves the limit of diagram \eqref{eq:Eobj}
        which defines $E_n$ in terms of all the permutations
        $\sigma\in\Perm$.  
\end{proposition}
\begin{proof}The proof is inspired by an analogous one in~\cite{DK20a} and hinges on the
fact that the contraction $\eps : E_n(V) \to V^{\otimes n}$ is split mono. To
prove the latter, let us consider the symmetrisation operator $\frac{1}{n!}
\sum_{\sigma \in \Perm} \sigma : V^{\otimes n} \to V^{\otimes
n}$~\cite{bourbaki98,comon08} --  it is a contraction because the
properties of norms entail,
\[
\left \lVert \frac{1}{n!} \sum_{\sigma \in \Perm} \sigma \right \rVert \leq
\frac{1}{n!} \sum_{\sigma \in \Perm}\norm{\sigma} = \frac{1}{n!}
\sum_{\sigma \in \Perm} 1 = 1
\]
It is then straightforward to show that this operator restricts on the codomain
to a linear map $\partial^n : V^{\otimes n} \to E_n(V)$ by taking advantage of
the fact that $\Perm$ is a group. Moreover $E_n(V)$ inherits its
norm from $V^{\otimes n}$ which yields $\norm{\partial^n} = \norm{\frac{1}{n!}
\sum_{\sigma \in \Perm} \sigma}$. Thus $\partial^n$ is a linear
contraction as well. Next, in order to prove that $\partial^n$ is a retraction of
$\eps$ consider the following facts. By construction we have $\sigma \comp \eps = \eps$
for all symmetries $\sigma\in\Perm$ which gives rise to the equation $\left
(\frac{1}{n!} \sum_{\sigma \in \Perm} \sigma \right ) \comp \eps =\partial^n\comp \eps= \id
\comp\, \eps$. Moreover $\eps$ is an inclusion. Therefore for every vector $v \in E_n(V)$
we obtain, 
\begin{flalign*} & \, \partial^n (\eps
(v)) = \textstyle{\frac{1}{n!} \sum_{\sigma \in \Perm}
\sigma } (\eps (v)) = \eps(v) = v 
\end{flalign*} 
The final step is to prove that every cone $f : U \to V^{\otimes n} \, \hat
\otimes_\pi \, W$ factorises uniquely through $\eps \otimes \id$. By composition
we obtain a linear contraction $(\partial^n \otimes \id) \comp f : U \to E_n(V)
\, \hat \otimes_\pi \, W$. Let us show that it factorises $f$ through $\eps
\otimes \id$. Consider a vector $u \in U$. By construction we know that
$f(u) = \sigma \otimes
\id \, (f(u))$ for all permutations $\sigma$ on $n$. This entails,
\begin{flalign*}
        & \, f(u) & \\
        & = \frac{1}{n!} \sum_{\sigma \in \Perm} \, \sigma \otimes \id \, 
        (f(u)) & \\
        & = \frac{1}{n!} \left ( 
                \sum_{\sigma \in \Perm} \, \sigma \right ) \otimes \id \, 
        (f(u)) & \text{\{Addition distributes over $(- \otimes \id)$\}}
        \\ 
        & = \left ( \frac{1}{n!} 
                \sum_{\sigma \in \Perm} \, \sigma \right ) \otimes \id \, 
        (f(u)) & \text{\{Scaling distributes over $(- \otimes \id)$\}}
        \\
        & = \partial^n \otimes \id \,  (f(u)) 
\end{flalign*}
We thus obtain the chain of equalities $(\eps \otimes \id) \comp (\partial^n \otimes
\id) (f(u))  = (\eps \otimes \id) (f(u)) = f(u)$.  Finally unicity follows
from the fact that $\eps$ is split mono.
\end{proof}

This yields a canonical $\N$-Lipschitz exponential comonad on $\Ban$, and we can interpret
$\sem{\g{n} \real}\defeq E_n\sem{\real}=E_n(\meas \R)$ where $\meas\R$ is the
Banach space of finite measures on $\R$. Note that the elements of $\sem{\g{n}
\real}$ are invariant under all permutations in $\Perm$, but need not in
general be \iid distributions. For example $\sem{\prog{replace}(k,m,n)}$
corresponds to the \iid case as it is given by the $k$-fold tensor of the
distribution $\mathrm{Bern}(\nicefrac{n}{(n+m)})$, but
$\sem{\prog{no\_replace}(k,m,n)}$ is permutation-invariant without being \iid
Quite a lot is know about permutation-invariant distributions like these,
usually known as \emph{finite exchangeable sequences} in the probabilistic
literature. In particular,~\cite{diaconis1980finite} shows
that the following metric axiom is sound.
\begin{flalign}
	\infer{\prog{replace}(k,m,n)=_{\nicefrac{4k}{(m+n)}}\prog{no\_replace}(k,m,n)}
	{} \label{eq:diaconis}
\end{flalign}

The denotation of $\prog{iid\_normal}(k;\mu,\sigma)$ is the linear, norm-1 operator defined by the Markov kernel $\R\times\R^+\to(\meas\R)^{\otimes n}\to\meas(\R^n), (\mu,\sigma)\mapsto\mathrm{Normal}(\mu,\sigma)^{\otimes n}$.
There is no known closed-form expression for the total variation distance between
Gaussian distributions. However, upper bounds are known. In particular,
following~\cite[Prop.~1.2]{devroye2018total}, we know that the metric axiom
below is sound.
\begin{flalign}
	\infer{\prog{iid\_normal}(k;\mu_1,\sigma_1)=_{\phi(\mu_1,\sigma_1,\mu_2,\sigma_2)}\prog{iid\_normal}(k;\mu_2,\sigma_2)}{}\label{eq:gaussians}
\end{flalign}
where $\phi(\mu_1,\sigma_1,\mu_2,\sigma_2)=\frac{1}{2}\sqrt{k\left(\frac{\sigma_2^2-\sigma_1^2+(\mu_1-\mu_2)^2}{\sigma_1^2}-\log\left(\frac{\sigma^2_2}{\sigma^2_1}\right)\right)}$. 

Now based on these axioms, and the metric equational rules of
Fig.~\ref{fig:theo_rules} we can easily bound the total variation distance
between the final position of two complex $k$-steps random walks of the type
used in Monte-Carlo simulations (\eg to value options~\cite{hull2003options}).
For example, consider first the random walk on $\R$ where at each step the
\emph{sign} of the jump is determined by a sample from $\prog{replace}(k,m,n)$
and its \emph{magnitude} by a sample from
$\prog{iid\_normal}(k;\mu_1,\sigma_1)$.  Suppose we want to bound the distance
of this walk with one whose sign is sampled from $\prog{no\_replace}(k,m,n)$
and magnitude from $\prog{iid\_normal}(k;\mu_2,\sigma_2)$ instead. Working
directly at the level of the semantics, this would be a highly non-trivial
task, however if we express these walks as programs in our graded system we can
straightforwardly compute such a bound. The walks can be programmed as follows:
\begin{flalign*}
	&\prog{walk1}\defeq\prog{pr}_{k,[1,1]}~\prog{replace}(k,m,n), \prog{iid\_normal}(k;\mu_1,\sigma_1)~\prog{fr}~x,y.~(2\ast \prog{dr}(x)-1)\ast\prog{dr}(y):~!_k\real
	\\
	&\prog{walk2}\defeq\prog{pr}_{k,[1,1]}~\prog{no\_replace}(k,m,n), \prog{iid\_normal}(k;\mu_2,\sigma_2)~\prog{fr}~x,y.~(2\ast\prog{dr}(x)-1)\ast\prog{dr}(y):~!_k\real
	\\
    &\prog{endpoint}(w)\defeq \prog{cp}_{(1,\ldots,1)}~w~\prog{to}~x_1,\ldots,x_n.~\prog{dr}(x_1)+\ldots+\prog{dr}(x_n):\real.
\end{flalign*}	
Using the metric axioms \eqref{eq:diaconis}-\eqref{eq:gaussians} and
Fig.~\ref{fig:theo_rules}, the bound can be straightforwardly checked to be
$\prog{endpoint(walk1)}=_{\nicefrac{4k}{(m+n)}+\phi(\mu_1,\sigma_1,\mu_2,\sigma_2)}\prog{endpoint(walk2)}$.
The higher-order features of the language would allow us to write the program
above more modularly by introducing an iterator and still reason
quantitatively about it. We chose the shorter, less modular presentation above
in the interest of brevity.

\section{Conclusions and future work}
We presented a sound and complete $\V$-equational system for a graded
$\lambda$-calculus via the notion of a Lipschitz exponential comonad.  We
showed how to build such comonads canonically via a universal construction and
applied our results to both timed and probabilistic computation.  There are
multiple research lines which we intend to explore next.  First, we believe
that the construction of Lipschitz exponential comonads is interesting
\emph{per se} and that it deserves further exploration from a more categorical
perspective.  For example, we are interested in knowing whether the adjunction
involved is monoidal and whether it arises from the development of general
results about graded (co)equational theories over (enriched) monoidal
categories.  Second, our results were applied to the setting of metric
equations only but they go beyond that -- in particular, we would like to
explore as well the inequational, ultra-metric, and fuzzy cases due to their
increasing relevance in the literature.  Third, whilst we presented relatively
straightforward metric equational theories and corresponding models for timed
and probabilistic computation, we are also interested in knowing whether the
same can be done for hybrid~\cite{neves18,goncharov20} and
quantum~\cite{nielsen02,neves22} computation, two rapidly emerging paradigms
with an intrinsically quantitative nature. Finally we are also interested in
knowing if there is any formal connection with previous work on the notion of
comonadic lax extension and the relational semantics involving modal
types~\cite{dal2022relational,abel2020unified}

\medbreak\noindent
\textbf{Acknowledgements.}
This work is financed by National Funds through FCT - Fundação para a Ciência e
a Tecnologia, I.P. (Portuguese Foundation for Science and Technology) within
project IBEX, with reference PTDC/CCI-COM/4280/2021.

\bibliographystyle{./entics}
\bibliography{notes}

\end{document}